\documentclass[a4 paper]{article}
\usepackage[utf8]{inputenc}
\usepackage[english]{babel}

\usepackage{comment}

\usepackage{setspace}
\usepackage[margin = 1.1in]{geometry} 

\usepackage[
	backend=biber,
	citestyle=numeric,
	bibstyle=numeric,
	sorting=none,
	maxbibnames=999]{biblatex}
\DeclareNameAlias{sortname}{family-given}
\usepackage{amsthm}
\usepackage{amsmath}	
\usepackage{amssymb}
\usepackage{mathtools}
\usepackage{multirow}
\usepackage{booktabs}
\usepackage{verbatim} 
\usepackage{slashed}
\usepackage{mathbbol}
\usepackage{verbatim}
\usepackage{ indentfirst}
\usepackage{microtype}
\usepackage{authblk}
\usepackage{mathrsfs}
\usepackage{environ}
\usepackage{mathdots}
\usepackage{csquotes}

\newcommand{\ie}{{\it i.e.}\ }

\usepackage{titlesec}
\titleformat{\chapter}[display]
  {\bfseries\Large}
  {\filright\MakeUppercase{\chaptertitlename} \Large\thechapter}
  {3.5ex}
  {\titlerule\vspace{1ex}\filleft}
  [\vspace{1ex}\titlerule]

\usepackage{caption}
\usepackage{subfig}

\numberwithin{equation}{section}

\NewEnviron{eqsplit}{%
\begin{equation}\begin{split}
  \BODY
\end{split}\end{equation}
}

\NewEnviron{eqsplit*}{%
\begin{equation*}\begin{split}
  \BODY
\end{split}\end{equation*}
}

\NewEnviron{subgather}{%
\begin{subequations}\begin{gather}
  \BODY
\end{gather}\end{subequations}
}

\NewEnviron{eq}{%
\begin{equation}
  \BODY
\end{equation}
}

\theoremstyle{plain}
\newtheorem{prp}{Proposition}[section]
\newtheorem{thrm}[prp]{Theorem}
\newtheorem{crll}[prp]{Corollary}
\newtheorem{lmm}[prp]{Lemma}

\theoremstyle{definition}
\newtheorem{dfn}[prp]{Definition}

\theoremstyle{remark}
\newtheorem*{remark}{Remark}

\DeclareMathOperator{\Tr}{Tr}

\newcommand{\id}{\mathbb{I}}

\newcommand{\Z}{\mathbb{Z}}

\newcommand{\N}{\mathbb{N}}
\newcommand{\R}{\mathbb{R}}
\newcommand{\C}{\mathbb{C}}
\def\cL{{\cal L}}
\newcommand{\Ham}{\mathcal{H}}
\newcommand{\Lag}{\mathscr{L}}
\newcommand{\Alg}{\mathcal{A}}

\newcommand{\cpb}[2]{\{\! | #1, #2| \! \}}

\newcommand{\pb}[2]{\{ #1, #2 \}}
\newcommand{\ip}[2]{ #1\lrcorner #2 }

\newcommand{\Omegaone}{\Omega^{(1)}}
\newcommand{\omegaone}{\omega^{(1)}}

\renewcommand{\epsilon}{\varepsilon}

\renewcommand{\tilde}{\widetilde}

\newcommand{\parder}[2]{\frac{\partial #1}{\partial #2}}

\newcommand{\tilpartial}{\tilde \partial}
\renewcommand{\phi}{\varphi}

\usepackage{hyperref}
\hypersetup{hidelinks}

\begin{document}
\title{\textbf{Multiform description of the AKNS hierarchy and classical r-matrix}}

\author{Vincent Caudrelier and Matteo Stoppato\footnote{Corresponding author \texttt{mmms@leeds.ac.uk}.}}
\date{ }

\maketitle
\begin{center}
    \vspace{-6ex}
	School of Mathematics, University of Leeds, LS2 9JT, UK  \\[3ex]
\end{center}

\vspace{0.5cm}
\begin{abstract}
    In recent years, new properties of space-time duality in the Hamiltonian formalism of certain integrable classical field theories have been discovered and have led to their reformulation using ideas from covariant Hamiltonian field theory: in this sense, the covariant nature of their classical $r$-matrix structure was unraveled. Here, we solve the open question of extending these results to a whole hierarchy. We choose the Ablowitz-Kaup-Newell-Segur (AKNS) hierarchy. To do so, we introduce for the first time a Lagrangian multiform for the entire AKNS hierarchy. We use it to construct explicitly the necessary objects introduced previously by us: a symplectic multiform, a multi-time Poisson bracket and a Hamiltonian multiform. Equipped with these, we prove the following results: $(i)$ the Lax form containing the whole sequence of Lax matrices of the hierarchy possesses the rational classical $r$-matrix structure; $(ii)$ The zero curvature equations of the AKNS hierarchy are multiform Hamilton equations associated to our Hamiltonian multiform and multi-time Poisson bracket; $(iii)$ The Hamiltonian multiform provides a way to characterise the infinite set of conservation laws of the hierarchy reminiscent of the familiar criterion $\{I,H\}=0$ for a first integral $I$. 
\end{abstract}

\section{Introduction}
 The seminal work of Gardner, Greene, Kruskal and Miura \cite{Gardner_Greene_Kruskal_Miura_1967} quickly followed by that of Zakharov and Shabat \cite{Shabat_Zakharov_1972} launched the modern era of integrable systems by introducing the Inverse Scattering Method (ISM). The work of Ablowitz, Kaup, Newell and Segur \cite{Ablowitz_Kaup_Newell_Segur_1974} generalised this method and introduced the idea of an integrable {\it hierarchy}. The breakthrough discovery of ISM provides one with a nonlinear analogue of the Fourier transform to solve an initial-value problem for particular classes of PDEs. The other breakthrough discovery, in \cite{Zakharov_Faddeev_1971} and \cite{Zakharov_Manakov_1974}, that these PDEs were also examples of infinite-dimensional Hamiltonian systems completely integrable in the sense of Liouville, established the other face of the coin of integrable systems: the ISM was also a means to obtain the map to action-angle variables in this infinite-dimensional setting. It also meant that these equations had a Lagrangian description and could be thought of as classical field theories. With this in mind, the classical $r$-matrix approach to classical integrable systems, initially introduced by Sklyanin \cite{Sklyanin_1979} to perform the canonical quantization of such systems into quantum field theories, quickly evolved into a theory of its own describing the Hamiltonian aspects of integrable systems while furnishing a geometric and algebraic framework in which the inverse scattering method could be reinterpreted (via the notion of Riemann-Hilbert factorization problems). It developed into a deep theory based on the work of Semenov-Tian-Shansky \cite{SemenovTianShansky_1983} and Drinfel'd \cite{Drinfeld_1983}. In the realm of classical field theories in $1+1$ dimensions, an integrable model comes as part of an entire hierarchy of commuting flows which can be viewed as the reason for its integrability, providing an infinite dimensional analog of the situation in classical mechanics giving rise to Liouville theorem, see e.g. \cite{Faddeev_Takhtajan_2007} for a detailed account of this point of view. \\
 
More recently, it became apparent that the Hamiltonian formulation of an integrable classical field theory, which breaks the natural symmetry between the independent variables enjoyed for instance by the Lax pair formulation and the zero curvature equation, represents a practical but also a conceptual limitation. Regarding the practical aspect, the limitation appeared in connection with integrable classical field theories in the presence of a defect \cite{Bowcock_Corrigan_Zambon_2004}, specifically when one tried to understand Liouville integrability of that context, see \cite{Caudrelier_2008, Avan_Doikou_2012} and references therein. This was explained in \cite{Caudrelier_Kundu_2015,Caudrelier_2015} where the idea of a ``dual'' Hamiltonian formulation of a given classical field theory was introduced to resolve the issue. An important observation, which was further developed in \cite{Avan_Caudrelier_Doikou_Kundu_2016}, is that in this dual formulation, both Lax matrices of the pair describing the model at hand possess the same classical $r$-matrix structure, pointing to a space-time duality of this structure. This was unlikely to be a coincidence and the question of the origin of this common structure led to the work \cite{Avan_Caudrelier_2017} where, in the case of the Ablowitz-Kaup-Newell-Segur (AKNS) hierarchy \cite{Ablowitz_Kaup_Newell_Segur_1974}, it was traced back to a Lie-Poisson bracket used by Flashka-Newell-Ratiu (FNR) in \cite{Flashka_Newell_Ratiu_1983}.\\
 
Regarding the conceptual limitation, this \enquote{dual} description was still unsatisfactory in the sense that one had to choose one of the independent variables or the other as a starting point, but it did not seem possible to construct a classical $r$-matrix formalism capable of including both independent variables simultaneously. This flaw of the (traditional) Hamiltonian formulation of a field theory has been known for a long time. This is what motivated the early work of De Donder and Weyl \cite{DeDonder_1930,Weyl_1935} aimed at establishing a {\it covariant Hamiltonian field theory}. A wealth of subsequent developments ensued, driven in particular by the desire to construct a covariant canonical quantization scheme. This area is too vast to review faithfully here and we only refer, somewhat arbitrarily and with apologies to missed authors, to \cite{Helein_2011} and the references in \cite{Caudrelier_Stoppato_2020}. As we showed in \cite{Caudrelier_Stoppato_2020}, some of the ideas in that area provide a solution to our conceptual problem: can we construct a covariant Poisson bracket such that the Lax form (sometimes also called Lax connection) possesses a classical $r$-matrix structure and such that the zero curvature equation can be written as a covariant Hamilton equation? This success in giving the $r$-matrix a covariant interpretation was illustrated explicitly on three models: the nonlinear Schr\"odinger (NLS) equation, the modified Korteweg-de Vries (mKdV) equation and the sine-Gordon model in laboratory coordinates. Each model was considered in its own right as an integrable classical field theory. However, as mentioned above, it is known that they come in hierarchies. Of those three models, two (NLS and mKdV) naturally fit into the AKNS hierarchy. We can now pose the problem that is addressed in the present work: how can we extend our results in \cite{Caudrelier_Stoppato_2020} to the whole AKNS hierarchy? In other words, denoting by $Q^{(k)}(\lambda)$, $k\ge 0$ the set of Lax matrices associated to the flows with respect to the variables $x^k$, $k\ge 0$ in the AKNS hierarchy, is it possible to construct a Poisson bracket which can take the Lax form $$
W(\lambda)=\sum_{k=0}^\infty Q^{(k)}(\lambda)\,dx^k
$$ as an argument and which exhibits the classical $r$-matrix structure known to exist for each individual $Q^{(k)}(\lambda)$ \cite{Avan_Caudrelier_2017}? Also, can we write the entire set of zero curvature equations of the hierarchy, \ie $$\partial_kQ^{(j)}(\lambda)-\partial_jQ^{(k)}(\lambda)+[Q^{(j)}(\lambda),Q^{(k)}(\lambda)]=0, \qquad j,k\ge 0$$
in Hamiltonian form? The appropriate generalisation of the covariant Poisson bracket and the covariant Hamiltonian used in \cite{Caudrelier_Stoppato_2020} to tackle this problem for individual classical field theories was introduced recently by the authors in \cite{Caudrelier_Stoppato_2020_2}. We call them {\it multi-time Poisson bracket} (for the obvious reason that it can be used to generate the flows with respect to all the ``times'' $x^k$ of the hierarchy) and {\it Hamiltonian multiform} respectively. The latter terminology comes from the fact that this object is derived from what is called a Lagrangian multiform, a notion introduced by Lobb and Nijhoff in \cite{Lobb_Nijhoff_2009} and which encodes integrability in a variational way. At first, Lagrangian multiforms appeared in the realm of fully discrete integrable systems, in order to encapsulate multidimensional consistency, which captures the analogue of the commutativity of Hamiltonian flows in continuous integrable systems. This original work stimulated further developments in discrete integrable systems \cite{Lobb_Nijhoff_Quispel_2009,Lobb_Nijhoff_2010,Bobenko_Suris_2010,Yoo-Kong_Lobb_Nijhoff_2011,Boll_Petrera_Suris_2014,Boll_Petrera_Suris_2015}, then progressively in continuous finite dimensional systems, see e.g. \cite{Suris_2013,Petrera_Suris_2017}, $1+1$-dimensional field theories, see e.g. \cite{Xenitidis_Nijhoff_Lobb_2011,Suris_2016,Suris_Vermeeren_2016,Vermeeren_2019,Sleigh_Nijhoff_Caudrelier_2019,Petrera_Vermeeren_2020}, and the first example in $2+1$-dimensions \cite{Sleigh_Nijhoff_Caudrelier_2019_2}. The proposed generalised variational principle produces the standard Euler-Lagrange equations for the various equations forming an integrable hierarchy as well as additional equations, originally called corner equations, which can be interpreted as determining the allowed integrable Lagrangians themselves. The set of all these equations is called multiform Euler-Lagrange equations. In \cite{Caudrelier_Stoppato_2020_2}, we realised that one could apply successfully the strategy we used in \cite{Caudrelier_Stoppato_2020} to obtain a covariant Poisson bracket and covariant Hamiltonian, which is based on the data of a single Lagrangian, to a Lagrangian multiform instead. The outcome is the desired multi-time Poisson bracket and Hamiltonian multiform. Our main results are:
\begin{enumerate}
    \item We introduce for the first time a Lagrangian multiform for the complete AKNS hierarchy, see equations \eqref{Lag_multi_eq1}-\eqref{Lag_multi_eq2}. This is achieved using a generating Lagrangian $\Lag(\lambda,\mu)$ which is a double formal series whose coefficients are the coefficients of the desired Lagrangian multiform. It is remarkable that we can produce a closed form expression in generating form while previous attempts involved complicated iterative procedures.
    \item Given this Lagrangian multiform $\Lag$, we obtain the multi-time Poisson bracket $\cpb{~}{~}$ and the Hamiltonian multiform $\Ham$ associated with it.
    \item We prove that the Lax form $\displaystyle W(\lambda)=\sum_{k=0}^\infty Q^{(k)}(\lambda)\,dx^k$ possesses a classical $r$-matrix structure with respect to the multi-time Poisson bracket above, \ie  $$\cpb{W_1(\lambda)}{W_2(\mu)} = [r_{12}(\lambda-\mu), W_1(\lambda)+W_2(\mu)],$$ where $r_{12}(\lambda)$ is the so-called rational $r$-matrix.  This shows, together with \cite{Caudrelier_Stoppato_2020} and \cite{Caudrelier_Stoppato_2020_2}, the potential of covariant Hamiltonian field theory to reproduce and generalise one of the crucial identities in integrable systems.
    \item We show the Hamiltonian multiform nature of the set of zero curvature equations $\partial_i Q^{(j)}(\lambda) - \partial_jQ^{(i)}(\lambda) + [Q^{(i)}(\lambda),Q^{(j)}(\lambda)] = 0$, $\forall i,j\ge 0$. We do this by proving that it is equivalent to the multiform Hamilton equations for the Lax form, \ie\footnote{Here and in the rest of the paper the notation $\displaystyle\sum_{i<j=1}^\infty $ means $\displaystyle \sum_{\substack{i,j=0\\i<j}}^\infty$.} 
    $$ 
    dW(\lambda) = \sum_{i<j=1}^\infty \cpb{H_{ij}}{W(\lambda)} \, dx^i \wedge dx^j,
    $$
    which take the suggestive form $\displaystyle dW(\lambda) = W(\lambda)\wedge W(\lambda)$ of the Maurer-Cartan equation. Moreover, we were able to reproduce the known series of conservation laws and conserved quantities of this hierarchy. This proves the success of Hamiltonian multiforms (and thus indirectly Lagrangian multiforms) in describing the main features of integrable systems.
\end{enumerate}

The paper is organised as follows. In Section \ref{AKNS section} we give a brief and not exhaustive description of the AKNS hierarchy and of the r-matrix theory. In Section \ref{sectionLag} we give the Lagrangian multiform description of the AKNS hierarchy, and we introduce the symplectic multiform. In Section \ref{Section_rmatrix} we treat the classical r-matrix structure of the multi-time Poisson bracket. In Section \ref{Section_Hammultiform} we give the Hamiltonian multiform description of the hierarchy, we show the Hamiltonian multiform nature of the zero curvature equations, and we obtain the familiar conservation laws in this framework. In Section \ref{Section_3times} we illustrate our results on the first three times of the hierarchy which include the nonlinear Schr\"odinger equation and the modified Korteweg-de Vries equation. Conclusions and the discussion of open problems are included in Section \ref{Conclusions section}. Appendix \ref{efcoordinates_section} contains some properties of the useful phase space coordinates $e(\lambda)$ and $f(\lambda)$ that we use to our advantage for many of our proofs. For a better flow, some of those long, and not necessarily illuminating, proofs are gathered in Appendix \ref{proof section}. 

\section{Algebraic construction of the AKNS hierarchy}\label{AKNS section}
In the 1983 paper \cite{Flashka_Newell_Ratiu_1983}, Flashka, Newell and Ratiu introduced an algebraic formalism to cast the soliton equations associated with the AKNS hierarchy into what is known as the Adler-Kostant-Symes scheme \cite{Adler_1979,Kostant_1979,Symes_1980}. At the same period, the russian school unraveled the structures underlying this type of construction which culminated in the classical $r$-matrix theory \cite{SemenovTianShansky_1983}, and the introduction of the notion of Poisson-Lie group \cite{Drinfeld_1983}. Here, we review some aspects of this topic, freely adapting and merging notations and notions coming from both sources. It had been known before \cite{Flashka_Newell_Ratiu_1983}, since the work of \cite{Ablowitz_Kaup_Newell_Segur_1974}, that the so-called AKNS hierarchy can be constructed by considering an auxiliary spectral problem of the form
\begin{eq}
\partial_x \psi = \begin{pmatrix}
-i \lambda &  q(x)\\
r(x) & i \lambda
\end{pmatrix}\psi\equiv P \psi\equiv (\lambda P_0+P_1)\psi\,,
\end{eq}
where
\begin{eq}
P_1 = \begin{pmatrix}
0& q\\
r & 0
\end{pmatrix}\,, \qquad P_0 = -i \sigma_3\,,
\end{eq}
as well as another equation of the form\footnote{Traditionally, the flows thus defined are associated to ``time'' variable $t^n$ but one of the main points of \cite{Flashka_Newell_Ratiu_1983} is that they all play the same role as $x$ which could be viewed as $t^1$ in this hierarchy. We simply denote them all by $x^n$ since whether they play the role of a space or time variable is really up to interpretation.}
\begin{eq}
\partial_{n}\psi =  Q^{(n)}(\lambda) \psi\,,\qquad \partial_n \coloneqq \parder{}{x^n}\,,
\end{eq}
with $Q^{(n)}(\lambda) = \lambda^n Q_0 + \lambda^{n-1} Q_1 + \dots + Q_n$ where each $Q_i$ is a $2\times 2$ traceless matrix. Then the compatibility condition $\partial_x \partial_{n} \psi = \partial_{n} \partial_x \psi$ translates into the well-known zero-curvature equation
\begin{eq}
\partial_{n} P(\lambda) - \partial_x Q^{(n)}(\lambda) + [P(\lambda),Q^{(n)}(\lambda)] = 0\,.
\end{eq}
By setting to zero every coefficient of powers of $\lambda$ one obtains a series of equations that allow to find $Q_0$, \dots, $Q_n$ recursively. This produces $Q_0 = P_0$, $Q_1 = P_1$ (up to some normalisation constants) and the entries of $Q_j$ with $j\ge 2$ are found to be polynomials in $q$, $r$ and their derivatives with respect to $x$. The last of these equations is
\begin{eq}
\partial_{n}{P_1} - \partial_x Q_n +[P_1,Q_n] = 0\,,
\end{eq}
and produces a partial differential equation for $q$ and $r$ viewed as functions of $x$ and $x^n$ which is integrable. Different values of $n$ gives the successive equations of the AKNS hierarchy. We list them for $n=0,1,2,3$, giving the name of the corresponding famous example (which is usually obtained by a further reduction, e.g. $r=\pm q^*$ for $n=2$ gives the (de)focusing nonlinear Schr\"odinger equation).
\begin{eq}\label{AKNS_FNR}
\begin{array}{ccr}
    q_0 = -2i q\,, & r_0 = 2i r & \mbox{Scaling}\,,\\
    q_1 = q_x\,,& r_1 = r_x  & \mbox{Translation }x \mapsto x + x^1\,,\\
    q_2 = \frac{i}{2}q_{xx} - i q^2 r \,,& r_2 = - \frac{i}{2}r_{xx} + i q r^2 & \mbox{Non-linear Schr\"odinger equation}\,,\\
    q_3 = -\frac{1}{4}q_{xxx} + \frac{3}{2} q r q_x  \,,& r_3 = - \frac{1}{4}r_{xxx} +\frac{3}{2}  q r r_x & \mbox{Modified Korteweg-de Vries equation}
    \end{array}
\end{eq}
It turns out that all these equations can be interpreted as Hamiltonian flows which commute with each other and can therefore be imposed simultaneously on the variable $q$ and $r$. This is ensured by that fact that the following zero curvature equations hold for any $k,n\ge 0$ (by setting $x=x^1$ and $Q^{(1)}=P$),
\begin{eq}
\partial_{n} Q^{(k)}(\lambda) - \partial_k Q^{(n)}(\lambda) + [Q^{(k)}(\lambda),Q^{(n)}(\lambda)] = 0\,.
\end{eq}
In \cite{Flashka_Newell_Ratiu_1983}, these facts and several others were cast into the algebraic setup of the Adler-Kostant-Symes scheme whereby one can introduce integrable Hamiltonian systems based on the decomposition of a Lie algebra into two Lie subalgebras which are isotropic with respect to an ad-invariant nondegenerate symmetric bilinear form on the Lie algebra. For the AKNS hierarchy, \cite{Flashka_Newell_Ratiu_1983} use the the Lie algebra $\cL$ of formal Laurent series in a variable $\lambda$ with coefficients in the Lie algebra $sl(2,\C)$, \ie the Lie algebra of elements of the form
\begin{eq}
X(\lambda)=\sum_{j=-\infty}^N X_j \lambda^j\,,~~X_j\in{\rm sl}(2,\C)\,;~~\text{for some integer $N$}\,,
\end{eq}
with the bracket given by
\begin{eq}
\label{Liebracket}
[X,Y](\lambda)=\sum_k\sum_{i+j=k}[X_i,Y_j]\lambda^k\,.
\end{eq}
There is a decomposition of $\cL$ into Lie subalgebras $\cL=\cL_-\oplus \cL_+$ where
\begin{eq}
\cL_-=\{\sum_{j=-\infty}^{-1} X_j \lambda^j\}\,,~~\cL_+=\{\sum_{j=0}^\infty X_j \lambda^j\,;~~X_j=0~~\forall j>N~~\text{for some integer $N\ge 0$}\}\,.
\end{eq}
This yields two projectors $P_+$ and $P_-$. The following ad-invariant nondegenerate symmetric bilinear form is used, for all $X,Y\in\cL$,
\begin{eq}
\label{bilinear_form}
\left(X,Y\right)=\sum_{i+j=-1}{\rm Tr}(X_i,Y_j)\equiv {\rm Res}_{\lambda}{\rm Tr}(X(\lambda)Y(\lambda))\,,
\end{eq}
Without entering the details of the construction, we present the summarised results of interest for us. The entire AKNS hierarchy can be obtained by considering $Q(\lambda)$ as the following formal series 
\begin{gather}
\label{defQ}
Q(\lambda) = \sum_{i=0}^\infty Q_i \lambda^{-i} = Q_0 +  \frac{Q_1}{\lambda}+ \frac{Q_2}{\lambda^2} + \frac{Q_3}{\lambda^3} + \dots\,,\\
Q_i = \begin{pmatrix}
a_i & b_i\\
c_i & -a_i
\end{pmatrix}\, \qquad a(\lambda) = \sum_i a_i \lambda^{-i}\,,\quad  b(\lambda) = \sum_i b_i \lambda^{-i}\,,\quad  c(\lambda) = \sum_i c_i \lambda^{-i}\,,
\end{gather}
and introducing the vector fields $\partial_n$ by
\begin{eq}\label{FNR_tNequations}
\partial_{n} Q(\lambda) = [P_+(\lambda^nQ(\lambda)), Q(\lambda)]=-[P_-(\lambda^nQ(\lambda)), Q(\lambda)]=[R(\lambda^nQ(\lambda)), Q(\lambda)]\,,
\end{eq}
where $R=\frac{1}{2}(P_+-P_-)$ is the endomorphism form of the classical $r$-matrix. It is well-known that this operator satisfies the modified classical Yang-Baxter equation 
and allows one to define a second Lie bracket $[~,~]_R$ on $\cL$ (see e.g. \cite{SemenovTianShansky_2008})
\begin{eq}
\label{R_bracket}
[X,Y]_R=[RX,Y]+[X,RY]\,.
\end{eq}
The significance of this reformulation is that the authors achieved several important results:
\begin{enumerate}
    \item The equations \eqref{FNR_tNequations} are commuting Hamiltonian flows associated to the Hamiltonian functions 
    \begin{eq}
\label{Casimirs}
h_k(X)=-\frac{1}{2}(S^k(X),X)\,,~~k\in\Z\,,~~(S^kX)(\lambda)=\lambda^kX(\lambda)\,.
    \end{eq} 
    which are Casimir functions with respect to the Lie-Poisson bracket associated to the Lie bracket \eqref{Liebracket}. As a consequence, these functions are in involution with respect to the Lie-Poisson bracket associated to the second Lie bracket \eqref{R_bracket} on $\cL$ and their Hamilton equations take the form of the Lax equation \eqref{FNR_tNequations};
    
    \item In this construction, one can get rid of the special role of the $x$ variable, which is now the variable $x^1$, no different from any of the other $x^n$. They then propose to define a hierarchy of integrable PDEs as follows: use \eqref{FNR_tNequations} for a fixed $n$ as a starting point to determined all the $Q_j$. This yields that $b_j$, $c_j$ for $j>n$ and $a_j$, $j>1$ are polynomials in $b_j,c_j$, $j=1,\dots,n$, which are now viewed as functions of $x^n$, and in their derivatives with respect to $x^n$. Then, one can use any one of the other variables $x^k$ to induce a Hamiltonian flow on the infinite dimensional phase space $b_j(x^n),c_j(x^n)$, $j=1,\dots,n$. The Hamilton equations takes the form of a zero curvature equation
\begin{eq}
\partial_kQ^{(n)}(x^n,\lambda)-\partial_nQ^{(k)}(x^n,\lambda)+[Q^{(n)}(x^n,\lambda),Q^{(k)}(x^n,\lambda)]=0
\end{eq}
where $Q^{(n)}(x^n,\lambda)$ denotes $P_+(\lambda^nQ(\lambda))$ where the above substitution for $a_j,b_j,c_j$ in terms of the finite number of fields $b_j(x^n),c_j(x^n)$, $j=1,\dots,n$ and their $x^n$ derivatives has been performed. See \cite{Avan_Caudrelier_2017} for more details about this.
    
    \item There exist generalised conservation laws $\parder{F_{jk}}{x^\ell} = \parder{F_{\ell k}}{x^j}$ for all $j,k,\ell\ge 0$ where $F_{kj}$ can be obtained efficiently from a generating function. For $j=1$, they reproduce the usual AKNS conservation laws with $F_{1k}$ being the conserved densities and $F_{\ell k}$ the corresponding fluxes.
\end{enumerate}
Those results are reviewed in detail in \cite{Avan_Caudrelier_2017} where the observation that one can start from an arbitrary flow $x^n$ is used to prove a general result on the $r$-matrix structure of dual Lax pairs which was first observed in \cite{Caudrelier_Kundu_2015} and \cite{Caudrelier_2015}. More explicitly, recall that the first $r$-matrix structure appeared in \cite{Sklyanin_1979} with the objective of quantizing the nonlinear Schr\"odinger equation. The fundamental observation is that one can encode the Poisson bracket $\{q(x),r(y)\}=i\delta(x-y)$ allowing to describe NLS as a Hamiltonian field theory into the (linear ultralocal) Sklyanin bracket \cite{Sklyanin_1979}
\begin{eq}
\label{SklyaninPB}
\{Q^{(1)}_1(x,\lambda)~,~Q^{(1)}_2(y,\mu)\}=\delta(x-y)\left[r_{12}(\lambda-\mu), Q^{(1)}_1(x,\lambda)+Q^{(1)}_2(y,\mu)\right]\,,
\end{eq}
where the indices $1$ and $2$ denote in which copy of the tensor product a matrix acts non trivially, e.g. $M_1=M\otimes\id$, $M_2=\id\otimes M$, and $r_{12}(\lambda)$ is the classical $r$-matrix related to the operator $R$ by
\begin{eq}
\forall X\in\cL\,,~~(RX)(\lambda)={\rm res}_{\mu}\Tr_2 (r_{12}(\lambda-\mu) X_{2}(\mu))\,,
\end{eq}
where $\Tr_2$ means that the trace is only taken over the second space in the tensor product. The main results of \cite{Avan_Caudrelier_2017} is that starting from a given $x^n$ and inducing the $x^k$ flow by considering the zero curvature equation 
\begin{eq}
\partial_k Q^{(n)}(x^n,\lambda)-\partial_n Q^{(k)}(x^n,\lambda)+[Q^{(n)}(x^n,\lambda),Q^{(k)}(x^n,\lambda)]=0
\end{eq}
or starting from $x^k$ and inducing the $x^n$ flow by considering the zero curvature equation 
\begin{eq}
\partial_nQ^{(k)}(x^k,\lambda)-\partial_kQ^{(n)}(x^k,\lambda)+[Q^{(k)}(x^k,\lambda),Q^{(n)}(x^k,\lambda)]=0
\end{eq}
yields the same set of integrable PDEs which possesses two Hamiltonian formulations with respect to two distinct Poisson brackets $\{~,~\}_n$ and $\{~,~\}_k$ for which one has
\begin{eq}
\{Q^{(n)}_1(x^n,\lambda)~,~Q^{(n)}_2(y^n,\mu)\}_n=\delta(x^n-y^n)\left[r_{12}(\lambda-\mu), Q^{(n)}_1(x^n,\lambda)+Q^{(n)}_2(y^n,\mu)\right]\,,
\end{eq}
and
\begin{eq}
\{Q^{(k)}_1(x^k,\lambda)~,~Q^{(k)}_2(y^k,\mu)\}_k=\delta(x^k-y^k)\left[r_{12}(\lambda-\mu), Q^{(k)}_1(x^k,\lambda)+Q^{(k)}_2(y^k,\mu)\right]\,.
\end{eq}

We want to stress that despite the deep observation that all independent variables $x^j$ play the same role, both in \cite{Flashka_Newell_Ratiu_1983} and \cite{Avan_Caudrelier_2017}, the authors still implement the step of using \eqref{FNR_tNequations} {\it first} for a fixed (but arbitrary) $x^n$ in order to produce a phase space for a {\it field theory} consisting of a finite number of fields $b_j(x^n)$, $c_j(x^n)$ $j=1,\dots,n$. This leads to a rather complicated construction of the Poisson brackets $\{~,~\}^n$ and $\{~,~\}^k$ in \cite{Avan_Caudrelier_2017} whose common $r$-matrix structure is traced back to the original Lie-Poisson bracket associated to the second Lie bracket \eqref{R_bracket}. This also points to the need of a truly covariant Poisson bracket capable of accommodating any pair of independent variables $x^n$ and $x^k$ {\it simultaneously} and producing an $r$-matrix structure for the associated {\it Lax form} $W(\lambda)=Q^{(n)}(\lambda)\,dx^n+Q^{(k)}(\lambda)\,dx^k$. This was achieved by us in \cite{Caudrelier_Stoppato_2020}. Another essential question was still pending: how to go beyond only a pair of times $x^n$ and $x^k$, corresponding to a single zero curvature equation, in order to include the entire hierarchy of flows? How to construct a Poisson structure capable of dealing with the corresponding Lax form $\displaystyle W(\lambda)=\sum_{j=0}^\infty Q^{(j)}(\lambda)\,dx^j$? 

In this paper, we answer these questions by avoiding altogether the first step of fixing a given time $x^n$ and by working with all the equations \eqref{FNR_tNequations} at once. They are interpreted as commuting Hamiltonian flows on a phase space with a countable number of coordinates $b_j$, $c_j$, $j\ge 1$. This interpretation is less known than the standard field theory viewpoint but it provides a deeper insight into the structure of the hierarchy. It is possible thanks to the very ideas behind a covariant approach to field theories which promote a {\it finite dimensional} phase space over a space-time manifold of dimension greater than one. Classical mechanics is viewed as a particular case where the space time manifold is reduced to $\R$ for a single time only. In our opinion, our interpretation is also a true implementation of the original observation that all independent variables $x^0$, $x^1$, $x^2\dots$ play a symmetric role. This is captured by our use of a Lagrangian and Hamiltonian multiform which do not distinguish any particular independent variable as being special.

Our main objective is to construct a Poisson bracket $\cpb{~}{~}$, called multi-time Poisson bracket, and a Hamiltonian multiform $\displaystyle H=\sum_{i<j=1}^\infty H_{ij}\,dx^i\wedge dx^j$ such that:
\begin{enumerate}
    \item It is possible to compute $\cpb{W_1(\lambda)}{W_2(\mu)}$ for the Lax form $\displaystyle W(\lambda)=\sum_{j=0}^\infty Q^{(i)}(\lambda)\,dx^j$ associated to the {\it entire} hierarchy, and to prove that it possesses the $r$-matrix structure;
    
    \item The collection of all the equations
    $\partial_{k} Q(\lambda) = [Q^{(k)}(\lambda), Q(\lambda)]$, $k\ge 0$ or, equivalently\footnote{This equivalence does not seem to be well-known but we use it all along and deal interchangeably with the FNR equations \eqref{FNR_tNequations} and the zero curvature equations \eqref{ZC}. The implication \eqref{FNR_tNequations}$\Rightarrow$\eqref{ZC} is shown for instance in \cite[Lemma 3.13]{Avan_Caudrelier_2017}. The converse is discussed in \cite[Chapter 5]{Newell_1985}.}, of all the zero curvature equations 
    \begin{eq}
    \label{ZC}
     \partial_{i} Q^{(j)}(\lambda)-\partial_{j}Q^{(i)}(\lambda)+[Q^{(j)}(\lambda), Q^{(i)}(\lambda)]=0\,,~~i,j\ge 0\,,
    \end{eq}
    can be written in Hamiltonian form as $\displaystyle dW(\lambda)=\sum_{i<j=1}^\infty\cpb{H_{ij}}{W(\lambda)}\,dx^i\wedge dx^j$.
    \end{enumerate}

It is remarkable that in order to achieve the first objective, one can apply the construction of \cite{Caudrelier_Stoppato_2020}, which was based on a single Lagrangian, to a Lagrangian multiform which contains in particular the collection of Lagrangians corresponding to all flows in the hierarchy. Then, using the construction introduced in \cite{Caudrelier_Stoppato_2020_2}, we are able to obtain the required Hamiltonian multiform. In our exposition, the use of generating functions in the form of formal (Laurent) series will turn out to be extremely efficient. With this in mind, 
we collect the following set of compatible Lax equations for $Q(\lambda)$ as defined in \eqref{defQ},
\begin{eq}
\partial_{k} Q(\lambda) = [Q^{(k)}(\lambda), Q(\lambda)]\,,~~k= 0,1,2,\dots\,,
\end{eq}
where $Q^{(k)}(\lambda)=P_+(\lambda^k Q(\lambda))$, into 
\begin{eq}\label{AKNS_generating}
D_\mu Q(\lambda) = \frac{[Q(\mu),Q(\lambda)]}{ \mu - \lambda}\,,
\end{eq}
where we introduced the derivation
\begin{eq}
\label{gen_der}
D_\mu \coloneqq \sum_{k=0}^\infty \frac{1}{\mu^{k+1}} \partial_k\,,
\end{eq}
and used the formal series identity
\begin{eq}
\sum_{k=0}^\infty \frac{Q^{(k)}(\lambda)}{\mu^{k+1}} = \frac{Q(\mu)}{\mu - \lambda}\,.
\end{eq}
It is important not to get confused by the notation $D_\mu$ which is not meant to be the partial derivative with respect to $\mu$, but simply the generating expression \eqref{gen_der}.
We remark that writing the AKNS hierarchy in the generating form \eqref{AKNS_generating} allows us to reproduce quickly known results. From the symmetry of the right-hand side in \eqref{AKNS_generating}, we have
\begin{eq}
D_\mu Q(\lambda) = D_\lambda Q(\mu)\,,
\end{eq}
which in component is
\begin{eq}
\partial_k Q_{j+1} = \partial_j Q_{k+1}\,,~~j,k\ge 0\,.
\end{eq}
Moreover, by means of the Jacobi identity we have
\begin{eq}
D_\lambda D_\mu Q(\nu) = D_\mu D_\lambda Q(\nu)\,,
\end{eq}
which means that the flows $\partial_j$ and $\partial_k$ commute\footnote{Of course, this had to be the case in the first place so as to allow us to consider those flows simultaneously and to define $D_\mu$, but this is a good check of the generating function formalism and an argument in favour of its efficiency.}. Finally, noting that the generating function of the Hamiltonian functions \eqref{Casimirs} is given by 
\begin{eq}
g(\lambda)\equiv-\frac{1}{2}\Tr{Q^2(\lambda)}=-\frac{1}{2}\Tr{Q^2_0}+\sum_{k=0}^\infty\frac{1}{\lambda^{k+1}}g_k\,,
\end{eq}
we find 
\begin{eq}
D_\mu g(\lambda)=0\,.
\end{eq}
This shows that the flows take place on the level surface $g(\lambda)=C(\lambda)$ where $C(\lambda)$ is a series in $\lambda^{-1}$ with constant coefficients. Therefore, in line with \cite{Flashka_Newell_Ratiu_1983}, we fix 
\begin{eq}
\Tr{Q^2(\lambda)}=-2\,,
\end{eq}
in the rest of this paper.

\section{Lagrangian and symplectic multiforms for the AKNS hierarchy}\label{sectionLag}
The framework of this work is the \emph{variational bi-complex}, whose algebraic description is giving in \cite{Dickey_2003}. Very roughly, let $M= \R^\N$ be the multi-time manifold with coordinates $(x_0,x_1,x_2,\dots)$, and let it be the base manifold of a fibred manifold, in which the section of the fibres are the fields and their derivatives $u^k_{(j)}$, $(j) = (j_0, j_1, j_2, \dots)$ where only a finite number of $j_i$'s are non zero. The base coordinates will be called \emph{horizontal}, and the fibred coordinates will be called \emph{vertical}. We will model the fibres using a differential algebra $\Alg$, with derivations $\partial_i$ acting in the usual way, $\partial_i u^k_{(j)} = u^k_{(j) +e_i}$. We will use two different differentials, a vertical one $\delta$, which resembles the variational derivative, and a horizontal one $d$, which is the usual total differential, such that $(\delta +d)^2=0$. With the symbol $\Alg^{(p,q)}$ we mean the set of forms with vertical degree $p$ and horizontal degree $q$ of the form
\begin{equation}
    \omega =\sum_{(i),(k),(j)} f^{(i)}_{(k),(j)}\delta u^{k_1}_{(i_1)} \wedge \dots \wedge \delta u^{k_p}_{(i_p)} \wedge dx^{j_1} \wedge \dots \wedge dx^{j_q}, \qquad f^{(i)}_{(k),(j)}\in \Alg\,.
\end{equation}
The reader can find more details on how to use these tools in the context of multiforms in \cite{Caudrelier_Stoppato_2020_2}.\\  

Lagrangian multiforms were first introduced by Lobb and Nijhoff in 2019 \cite{Lobb_Nijhoff_2009} to describe integrable hierarchies using a variational principle, in a way that preserves multidimensional consistency. For $1+1$-dimensional field theories, the starting point is to consider a two-form
\begin{eq}
\Lag[u] = \sum_{i<j=1}^\infty L_{ij}[u] \,dx^{ij}\,.
\end{eq}
The fields $u$ themselves are functions of the multi-time variables $x^0,x^1,x^2,\dots$ For each $i,j$, the notation $L_{ij}[u]$ means that $L_{ij}$ is a function of a finite set of fields collectively denoted by $u$ for simplicity, and their derivatives with respect to the multi-time variables up to some finite order. We used the notation $dx^{ij} = dx^i \wedge dx^j$, and the convention $L_{ij}[u] = - L_{ji}[u]$. In this paper, the coefficients $L_{ij}$ do not depend explicitly on the multi-time variables. The equations of the hierarchy, together with the so called \emph{corner equations}, are then obtained by setting
\begin{eq}
\delta d \Lag[u] = 0 \,,
\end{eq}
that are called \emph{multiform Euler-Lagrange equations}.  These are obtained via a generalised variational principle for an action associated to $\Lag[u]$. The reader is referred to \cite{Sleigh_Nijhoff_Caudrelier_2019_2, Vermeeren_2018} for details. The main sign of integrability put forward in the original work \cite{Lobb_Nijhoff_2009} was the \emph{closure relation}. In our context, it means that $d \Lag =0$ ``on-shell''  \ie when the multiform Euler-Lagrange equations $\delta d \Lag = 0$ hold. We include the closure relation as a property of a Lagrangian multiform in the present work.\\

In \cite{Caudrelier_Stoppato_2020_2} we introduced and developed the Hamiltonian counterpart called \emph{Hamiltonian multiforms}, which for $1+1$-field theories are two-forms
\begin{eq}
\Ham[u] = \sum_{i<j=1}^\infty H_{ij}[u] \,dx^{ij}\,.
\end{eq}
When paired with a \emph{symplectic multiform} $\Omega \in \Alg^{(2,1)}$, we can obtain the so-called \emph{multiform Hamilton equations} as
 \begin{eq}
\label{Ham_eqs}
\delta \Ham = \sum_{i=0}^\infty dx^i \wedge \ip{\tilpartial_i}\Omega\,,~~\text{where}~~\tilpartial_i=\sum_{(k)} u_{(k) + e_i} \parder{}{u_{(k)}}\,.
\end{eq}
We remark that for a given hierarchy both $\Ham$ and $\Omega$ can be obtained from the Lagrangian multiform $\Lag$ through a \emph{Legendre transform}-type procedure \cite[Definition 2.3]{Caudrelier_Stoppato_2020_2}. For a specific class of horizontal forms called \emph{Hamiltonian forms}, we can introduce the multiform analog of a Poisson structure, that we called \emph{multi-time Poisson bracket}: for any two Hamiltonian forms $F,G$ we define
\begin{eq}
 \cpb{F}{G} = (-1)^r\ip{\xi_F}{\delta G}\,,
\end{eq}
where $r$ is the horizontal degree of $F$ and $\xi_F$ is the (multi)vector field satisfying $\ip{\xi_F}{\Omega} = \delta F$. We refer the reader to the paper \cite{Caudrelier_Stoppato_2020_2} for an introduction and detailed explanation of these concepts.\\


\subsection{Lagrangian multiform}

We now introduce a Lagrangian multiform which allows us to implement the strategy \cite{Caudrelier_Stoppato_2020_2} recalled briefly above and obtain $\Ham$ and $\Omega$ for the AKNS hierarchy. 
Recall that the collection of flows in the AKNS hierarchy is written in generating form as 
\begin{subgather}
\label{eq_Q}
    D_\mu Q(\lambda) =\frac{[Q(\mu),Q(\lambda)]}{\mu-\lambda}\,, \qquad Q(\lambda) = \sum_{i=0}^\infty \frac{Q_i}{\lambda^i}\,,\\
    Q(\lambda) = \begin{pmatrix}
    a(\lambda) & b(\lambda)\\
    c(\lambda) & -a(\lambda)
    \end{pmatrix}\,,\qquad  Q_i = \begin{pmatrix}
    a_i & b_i\\
    c_i & -a_i
    \end{pmatrix}\,,\\
    \frac{1}{2}\Tr Q(\lambda)^2 = a^2(\lambda) + b(\lambda) c(\lambda) = -1\,.
\end{subgather}
where $\lambda$ and $\mu$ are formal parameters. In order to find an appropriate Lagrangian multiform, it is convenient to note that we can write $Q(\lambda)$ as
\begin{eq}
Q(\lambda) = \phi(\lambda) Q_0 \phi(\lambda)^{-1}
\end{eq}
with $Q_0= - i \sigma_3$ being constant and 
\begin{eq}
\phi(\lambda)=\id+\sum_{j=1}^\infty\frac{\phi_j}{\lambda^j}\,.
\end{eq} 
This has been established independently from various angles, in relation to the factorization theorem, see e.g. \cite{SemenovTianShansky_2008} or in relation to vertex operators, see e.g. \cite[Chapter 5]{Newell_1985}. Contrary to the parametrization used in the latter book, we find it useful to use the following remarkable set of coordinates found in \cite{Flashka_Newell_Ratiu_1983}
\begin{eq}
e(\lambda)  = \frac{b(\lambda)}{\sqrt{i-a(\lambda)}}= \sum_{i=1}^\infty \frac{ e_i}{\lambda^i}\,,~~f(\lambda) = \frac{c(\lambda)}{\sqrt{i-a(\lambda)}} = \sum_{i=1}^\infty \frac{f_i}{\lambda^i}, \qquad \text{(note: }e_0 = f_0 = 0),
\end{eq}
 and set
\begin{eq}
\label{form_phi}
\phi(\lambda) = \frac{1}{\sqrt{2i}} \begin{pmatrix}
\sqrt{2i-e(\lambda)f(\lambda)}&e(\lambda)\\
-f(\lambda) & \sqrt{2i-e(\lambda)f(\lambda)}
\end{pmatrix}\,.
\end{eq}
A direct calculation using $a^2(\lambda)+b(\lambda)c(\lambda)=-1$ shows that $\det \phi(\lambda)=1$ and $ \phi(\lambda) (-i\sigma_3) \phi(\lambda)^{-1}=Q(\lambda)$ as required.
The reader can find more about the coordinates $e(\lambda)$, $f(\lambda)$ in Appendix \ref{efcoordinates_section}. Their main property is that they provide Darboux coordinates for all the single-time Poisson brackets $\pb{~}{~}_k$ constructed from the single-time symplectic forms $\omega_k$, see Corollary \ref{corollary darboux coordinates} and Proposition \ref{proposition decomposition} below. We can now formulate the first main result of this section.
 We obtain the desired Lagrangian multiform $\displaystyle \Lag=\sum_{i<j=1}^\infty L_{ij}\,dx^{ij}$ using the generating function formalism and collecting the coefficients $L_{ij}$ into a formal series in $\lambda^{-1}$ and $\mu^{-1}$ as follows
\begin{eq}
\Lag(\lambda,\mu) = \sum_{i,j=0}^\infty \frac{L_{ij}}{\lambda^{i+1} \mu^{j+1}}\,.
\end{eq}
By a slight abuse of language, we will also call $\Lag(\lambda,\mu)$ a Lagrangian multiform.
\begin{thrm}[Lagrangian multiform and multiform Euler-Lagrange equations]
\label{Lag_multi}
Define
\begin{eq}\label{Lag_multi_eq1}
\Lag(\lambda,\mu) = K(\lambda,\mu) - V(\lambda,\mu)\,,
\end{eq}
where
\begin{equation}\label{Lag_multi_eq2}
    K(\lambda,\mu) = \Tr{\left( \phi(\mu)^{-1} D_\lambda \phi(\mu)Q_0 - \phi(\lambda) ^{-1} D_\mu \phi(\lambda)Q_0     \right) }\,, \qquad V(\lambda,\mu)= - \frac{1}{2} \Tr{ \frac{(Q(\lambda)-Q(\mu))^2}{\lambda-\mu} }\,.
\end{equation}
Then $\Lag(\lambda,\mu)$ is a Lagrangian multiform for the AKNS hierarchy equations \eqref{eq_Q}. 

\end{thrm}
Indeed, the multiform Euler-Lagrange equations $\delta d \Lag =0$ are given by
\begin{eq}
D_\mu Q(\lambda) = \frac{[Q(\mu),Q(\lambda)]}{\mu-\lambda}\,,
\end{eq}
and the closure relation $d\Lag=0$ is satisfied on those equations. In generating form, the latter is equivalent to 
\begin{equation}
    D_\nu \Lag(\lambda,\mu) + D_\lambda \Lag(\mu,\nu) + D_\mu \Lag (\nu,\lambda) =0.
\end{equation}
The proof is given in Appendix \ref{Proof_Lag_multi}.
\begin{remark}
Although we discovered it differently, the Lagrangian multiform \eqref{Lag_multi_eq1}-\eqref{Lag_multi_eq2} bears some striking resemblance with the Zakharov-Mikhailov (ZM) Lagrangian appearing in \cite{Zakharov_Mikhailov_1980}, despite the fact that the latter is a standard Lagrangian and not a multiform. The ZM Lagrangian was introduced to provide a variational description of the system of compatibility conditions (zero curvature equations) corresponding to a Lax pair of matrices which are rational functions of the spectral parameter with distinct simple poles. Interestingly, one can formally relate our Lagrangian multiform to the ZM Lagrangian but we do not elaborate on this here as we have not gained any insight on either our results or the ZM Lagrangian by doing so. A Lagrangian multiform constructed on the ZM Lagrangian was obtained in  \cite{Sleigh_Nijhoff_Caudrelier_2019} and used to obtain a variational derivation of Lax pair equations themselves. In that same paper, the authors presented the first few coefficients of the Lagrangian multiform for the AKNS hierarchy but it was not clear how these derive directly from the ZM Lagrangian multiform. Our Lagrangian multiform and Theorem \ref{Lag_multi} fill in this gap and provide the complete set of coefficients $L_{ij}$ of the Lagrangian multiform for the AKNS hierarchy. We note that Lagrangians producing the zero curvature equations \eqref{ZC} in potential form were obtained in \cite{Nijhoff_1986}. They involved a potential function denoted by $H$ in that paper which produces the Lax matrices $Q^{(k)}$ we use here via the relation $Q^{(k)}=\partial_{k-1}H$. However, assembling all those Lagrangians into a two-form does not seem to provide a Lagrangian multiform for the set of AKNS equations. The closure relation does not hold for instance.  
\end{remark}

To help the reader recognize the most familiar models, we write some of the coefficients of the Lagrangian multiform explicitly using our formula. Using the expansion $\displaystyle\Lag(\lambda,\mu) = \sum_{i<j=1}^\infty L_{ij}/\lambda^{i+1}\mu^{j+1}$ we have, for all $i,j\ge 0$
\begin{eq}
    L_{ij} = \frac{1}{2}\sum_{k=1}^j (f_k \partial_i e_{j+1-k} - e_k \partial_i f_{j+1-k}) - \frac{1}{2} \sum_{k=1}^i(f_k \partial_j e_{i+1-k} - e_k \partial_j f_{i+1-k}) - V_{ij}\,.
\end{eq}
where the coefficients $V_{ij}$ are given by
 \begin{eq}
V_{ij} = \sum_{k=0}^i (2 a_k a_{i+j+1-k} + b_k c_{i+j+1-k} + c_k b_{i+j+1-k} )\,.
\end{eq} 
Recall that the elements $a_j$, $b_j$ and $c_j$ can all be expressed in terms of the coordinates $e_j$ and $f_j$ (see Appendix \ref{efcoordinates_section}). At this stage, no particular choice of time has been made to write these Lagrangians as field theory Lagrangian, in the spirit of \cite{Avan_Caudrelier_Doikou_Kundu_2016} for instance. Hence, as an example, we simply have
\begin{eq}
    L_{12} = \frac{1}{2} (f_1 \partial_1 e_2 - e_1 \partial_1 f_2 + f_2 \partial_1 e_1 - e_2 \partial_1 f_1) - \frac{1}{2} (f_1 \partial_2 e_1 - e_1 \partial_2 f_1) - V_{12}\,,
\end{eq}
and
\begin{eq}
    L_{13} = \frac{1}{2} (f_1 \partial_1 e_3 - e_1 \partial_1 f_3 + f_2 \partial_1 e_2 - e_2 \partial_1 f_2 + f_3 \partial_1 e_1 - e_3 \partial_1 f_1) - \frac{1}{2} (f_1 \partial_3 e_1 - e_1 \partial_3 f_1) - V_{13}\,,
\end{eq}
which produce partial differential equations for the phase space coordinates $e_j$, $f_j$, $j=1,2,3$. 
Now to make contact with the more familiar form of these Lagrangians and the corresponding equations of motion, we express the phase space coordinates in terms of $b_1=q$, $c_1=r$ and their $x^1$ derivatives\footnote{The reader can find the relations between the $e_i$'s and $f_i$'s and $q$ and $r$ and their derivative with respect to $x^1$ in Appendix \ref{efcoordinates_section}.}. Note that this amounts to choosing the $x^1$ equation in \eqref{FNR_tNequations} and use it to {\it solve} for $Q_j$ (standard field theory point of view). Doing so yields,
\begin{eq}
\label{NLS_Lag}
 L_{12} = \frac{i}{4} (q_2r - qr_2 ) + \frac{1}{8} (r q_{11} + q r_{11}) - \frac{1}{4} q^2 r^2\,,
\end{eq}
and
\begin{eq}
 L_{13} = \frac{i}{4}(r q_3 - q r_3) + \frac{i}{16} ( q_{111} r - q r_{111}) + \frac{3i}{16} qr (qr_1 - r q_1)\,,
\end{eq}
which are known Lagrangians whose Euler-Lagrange equations are
\begin{subgather}
\label{eq_L12}
iq_2 + \frac{1}{2}q_{11} -  q^2 r = 0\,,\qquad i r_2 - \frac{1}{2}r_{11} + q r^2\,,\\
\label{eq_L13}
q_3 + \frac{1}{4}q_{111} - \frac{3}{2}qrq_1 = 0\,,\qquad r_3 + \frac{1}{4}r_{111} - \frac{3}{2}qrr_1 = 0\,.
\end{subgather}
These are the (unreduced) NLS and mKdV systems respectively. We can just as easily produce the Lagrangian $ L_{23}$, first in the $e$ and $f$ coordinates and then, if desired, in the $q$ and $r$ coordinates as before. It reads
\begin{eqsplit}
\label{eq_L23}
    L_{23} =& \frac{i}{16}(r q_{112}  - qr_{112} )+ \frac{i}{16}(q_1 r_{12} - q_{12} r_1) - \frac{i}{16}(q_{11} r_2 - q_2 r_{11})\\
    &- \frac{3i}{16}qr(r q_2 - qr_2)- \frac{1}{8}(q_{13}r + q r_{13}) + \frac{1}{8}(r_1 q_3 + q_1 r_3)\\
    & + \frac{1}{16} q_{11} r_{11} - \frac{qr}{8}(q r_{11} + q_{11} r) + \frac{1}{16}(q r_1 - q_1 r)^2 + \frac{1}{4} q^3 r^3\,,
\end{eqsplit}
and its Euler-Lagrange equations are just consequence of \eqref{eq_L12}-\eqref{eq_L13}.
\begin{remark}
The partial Lagrangian multiform thus derived here for the first three times $L_{12}\, dx^{12} + L_{23}\, dx^{23} + L_{13}\, dx^{13}$ is equivalent to the one first obtained in \cite{Sleigh_Nijhoff_Caudrelier_2019}, up to an overall coefficient $\frac{i}{2}$ and the (total) differential of $- \frac{1}{8}(q_1 r + r_1 q)\, dx^2 - \frac{i}{16} (q_{11}r - q r_{11}) \, dx^3$. 
\end{remark}

\subsection{Symplectic multiform}
Equipped with a Lagrangian multiform for the AKNS hierarchy, we now construct the associated symplectic multiform $\Omega$. As always, it is very convenient to work with generating functions so we introduce 
\begin{eq}
\Omegaone(\lambda)=\sum_{j=0}^\infty\frac{\omegaone_j}{\lambda^{j+1}}\,,~~\Omega(\lambda)=\sum_{j=0}^\infty\frac{\omega_j}{\lambda^{j+1}}\,,
\end{eq}
to represent respectively 
\begin{eq}
\Omegaone=\sum_{j=0}^\infty\omegaone_j\wedge dx^j\,,~~\Omega=\sum_{j=0}^\infty\omega_j\wedge dx^j\,.
\end{eq}
As before, by a slight abuse of language, we also call $\Omega(\lambda)$ symplectic multiform. 
\begin{prp}\label{prop_symplecticmultiform}
The symplectic multiform associated to $\Lag(\lambda,\mu)$ is given by 
\begin{eq}
\label{form_Omega}
\Omega(\lambda) = - \Tr\left( Q_0 \phi(\lambda)^{-1} \delta \phi(\lambda) \wedge \phi(\lambda)^{-1} \delta \phi(\lambda)\right)\,.
\end{eq}
\end{prp}
The proof is in Appendix \ref{proof_symplecticmultiform}.
\begin{remark}
The expression for $\Omega(\lambda)$ is reminiscent of the well-known expression for the (pull-back to the group of the) Kostant-Kirillov symplectic form on a coadjoint orbit of the loop algebra $\cL$ through the element $Q_0$. To make this more precise, let us use for instance the formulas in \cite[Section 3.3]{Babelon_Bernard_Talon_2003} giving the expression of the pull-back to the group of the Kostant-Kirillov form for the orbit through a diagonal matrix polynomial $A(\lambda)$,
\begin{eq}
\omega={\rm Res}_{\lambda}\Tr\left(A(\lambda)g^{-1}(\lambda)\delta g(\lambda)\wedge g^{-1}(\lambda)\delta g(\lambda)\right)\,.
\end{eq}
Here, choosing $A(\lambda)=-i\lambda^k\sigma_3$, $k\ge 0$, and $g(\lambda)=\phi(\lambda)$, we get the connection between our symplectic multiform and the Kostant-Kirillov form
\begin{eq}
\omega={\rm Res}_{\lambda}\lambda^k\Omega(\lambda)=\omega_k\,.
\end{eq}
In particular, each single-time symplectic form $\omega_k$ corresponds to $\omega$ on the orbit of the element $-i\lambda^k\sigma_3$. Therefore, our symplectic multiform contains in a single object all those symplectic forms. This is the first time such an object is derived and, to our knowledge, it is the first time that a Kostant-Kirillov symplectic form is derived from a Lagrangian perspective.
\end{remark}
As a consequence of the explicit formula for $\Omega$, we get the following remarkable result that the $e,f$ coordinates provide Darboux coordinates.
\begin{crll}\label{corollary darboux coordinates}
\begin{eq}
\Omega(\lambda) = \delta f(\lambda) \wedge \delta e(\lambda)\,,
\end{eq}
and hence, $\omega_0=0$ and,
\begin{eq}
\label{form_omegak}
\omega_k = \sum_{i=1}^k \delta f_i \wedge \delta e_{k+1-i}\,,~~\forall k\ge 1\,.
\end{eq}
\end{crll}
\begin{proof}
Direct calculation by inserting \eqref{form_phi} into \eqref{form_Omega}.
\end{proof}

\section{Classical r-matrix structure}\label{Section_rmatrix}

\subsection{Hamiltonian forms and multi-time Poisson bracket}
As explained in \cite{Caudrelier_Stoppato_2020_2}, only for a specific class of horizontal forms called \emph{Hamiltonian forms}, \ie a form $F$ such that there exist a (multi)vector field $\xi_F$ (called \emph{Hamiltonian vector field}) that satisfies the relation $\ip{\xi_F}{\Omega} = \delta F$. Having the symplectic multiform $\Omega$ at our disposal, we can investigate in detail under which conditions a horizontal form is Hamiltonian and then compute the multi-time Poisson bracket for two such forms. Recall (see \cite{Caudrelier_Stoppato_2020_2}) that in our case, only $0$- and $1$-forms can be non-trivial Hamiltonian forms. We have the following two propositions, the proofs of which are given in Appendix \ref{proof_Hamiltonian1forms} and \ref{proof_Hamiltonian0forms}.
\begin{prp}\label{Hamiltonian1forms}
A $1$-form $\displaystyle F = \sum_{k=0}^\infty F_k \, dx^k$ is Hamiltonian with respect to $\Omega$ if and only if $F_0$ is constant and, for all $k\ge 1$, $F_k$ depends only on the coordinates $(e_1,\dots,e_k,f_1,\dots,f_k)$ and 
\begin{equation}
   \parder{F_k}{e_j} = \parder{F_{k+1}}{e_{j+1}}\,,\qquad \parder{F_k}{f_j} = \parder{F_{k+1}}{f_{j+1}}\,,~~j=1,\dots,k\,.
\end{equation}
Its Hamiltonian vector field is given by
\begin{equation}
    \xi_F = \sum_{k=1}^\infty \left( - \parder{F_k}{f_1} \partial_{e_k} + \parder{F_k}{e_1} \partial_{f_k} \right)\,.
\end{equation}
\end{prp}

\begin{prp}\label{Hamiltonian0forms}
Every 0-form $H(e_1, \dots, f_1, \dots)$ is Hamiltonian with respect to $\Omega$, with Hamiltonian vector field given by
\begin{equation}
\label{xi_H}
    \xi_H = \sum_{i=1}^\infty\left( - \parder{H}{f_i} \partial_{e_1} \wedge \partial_i +  \parder{H}{e_i} \partial_{f_1} \wedge \partial_i \right)\,.
\end{equation}
\end{prp}
Note that in practice, we will deal with $0$-forms that depend only on a finite number of coordinates $e_j$, $f_j$ in which case the sum in \eqref{xi_H} truncates accordingly.\\

We can now construct the multi-time Poisson bracket with respect to $\Omega$ between two Hamiltonian forms $F$ and $G$ as
\begin{eq}
 \cpb{F}{G} = (-1)^r\ip{\xi_F}{\delta G}
\end{eq}
where $r$ is the horizontal degree of $F$. Theorem 2.18 in \cite{Caudrelier_Stoppato_2020_2} gives the decomposition of the multi-time Poisson brackets in terms of the single-time Poisson brackets $\pb{~}{~}_k$ constructed from the symplectic forms $\omega_k$: for every $F$, $G$ functions of $e_1,\dots,e_k, f_1,\dots,f_k$ we have $\pb{F}{G}_k = - \ip{\xi_F}{\delta G}$ where $\ip{\xi_F}{\omega_k} = \delta F$. Given that we know the explicit form of the single-time symplectic forms $\omega_k$, see \eqref{form_omegak}, we obtain the following specialisation as a consequence.
\begin{prp}[Decomposition of the multi-time Poisson brackets]\label{proposition decomposition}
The multi-time Poisson brackets with respect to $\Omega$ of two Hamiltonian $1$-forms $\displaystyle F=\sum_{k=0}^\infty F_k \, dx^k$ and $\displaystyle  G = \sum_{k=0}^\infty G_k \, dx^k$ satisfies the following decomposition:
\begin{equation}
    \cpb{F}{G} = \sum_{k=0}^\infty \pb{F_k}{G_k}_k \, dx^k\,, \qquad  \pb{F_k}{G_k}_k = \sum_{j=1}^k \left( \parder{F_k}{f_j}\parder{G_k}{e_{k-j+1}} - \parder{F_k}{e_j} \parder{G_k}{f_{k-j+1}}\right),\quad k \ge 1 
\end{equation}
and $\pb{F_0}{G_0}_0=0$.
\end{prp}
Thanks to the the propositions above, we can prove by direct but long calculations that the multi-time Poisson bracket $\cpb{~}{~}$ satisfies the Jacobi identity.
\begin{prp}[Jacobi identity]\label{prp_Jacobiidentity}
If $F,G,K \in \Alg^{(0,1)}$ and $H \in \Alg$ are Hamiltonian forms, we have that
\begin{enumerate}
    \item $\cpb{F}{G}$ and $\cpb{F}{H}$ are respectively a Hamiltonian 1-form and a Hamiltonian 0-form,
    \item $\cpb{\cpb{F}{G}}{K} + \cpb{\cpb{K}{F}}{G} + \cpb{\cpb{G}{K}}{F} =0$,
    \item $\cpb{\cpb{F}{G}}{H} + \cpb{\cpb{H}{F}}{G} + \cpb{\cpb{G}{H}}{F} =0$.
\end{enumerate}
\end{prp}
\begin{remark}
It is known (see e.g. \cite{Forger_Salles_2015}) that the Jacobi identity is not necessarily satisfied by a covariant Poisson bracket. This problem could therefore be present in general for a multi-time Poisson bracket (which can be viewed as a generalisation of a covariant Poisson bracket). This is why the Jacobi identity was not discussed in \cite{Caudrelier_Stoppato_2020_2} and why we checked it here directly.
\end{remark}

\subsection{Classical r-matrix structure of the multi-time Poisson bracket}
\begin{dfn}
We call Lax form the following horizontal $1$-form with matrix coefficient
\begin{eq}
W (\lambda) =\sum_{i=0}^\infty Q^{(i)}(\lambda) \, dx^i
\end{eq}
where, for $i\ge 0$, 
\begin{eq}
Q^{(i)}(\lambda)\coloneqq P_+(\lambda^iQ(\lambda))\,.
\end{eq}
\end{dfn}
Note that the definition of a Hamiltonian form extends naturally to the case of matrix coefficients by requiring that each entry be a Hamiltonian form. We also extend the tensor notation used in the Sklyanin bracket, as reviewed in Section \ref{AKNS section}, to the present situation as follows
\begin{gather}
W_1(\lambda)\equiv \sum_{i=0}^\infty Q^{(i)}(\lambda)\otimes\id \, dx^i=\sum_{i=0}^\infty\,\sum_{k=+,-,3} Q^{(i)}_k(\lambda)\,\sigma_k \otimes\id \, dx^i \equiv \sum_{k=+,-,3} W^k(\lambda) \sigma_k \otimes \id \\
W_2(\lambda)\equiv \sum_{i=0}^\infty \id\otimes Q^{(i)}(\lambda) \, dx^i =  \sum_{i=0}^\infty\,\sum_{k=+,-,3}Q^{(i)}_k(\lambda)\, \id\otimes \sigma_k  \, dx^i \equiv \sum_{k=+,-,3}W^k(\lambda) \id \otimes \sigma_k\,,
\end{gather}
where we have written the matrices in terms of the $sl(2,\mathbb{C})$ basis $\sigma_+$, $\sigma_-$, $\sigma_3$. We define the multi-time Poisson bracket between $W_1(\lambda)$ and $W_2(\mu)$ by
\begin{eq}
 \cpb{W_1(\lambda)}{W_2(\mu)} = \sum_{k,\ell=+,-,3} \cpb{W^k(\lambda)}{W^\ell (\mu)} \,\sigma_k \otimes \sigma_\ell\,.
\end{eq}
Finally, we define the commutator of a matrix $0$-form $M$ and a matrix $1$-form $W$ by
\begin{eq}
[M,W]\equiv \sum_{i=0}^\infty[M,W_i]\,dx^i\,.
\end{eq}
We are now ready to formulate the main result of this section, the proof of which is long but straightforward and is given in Appendix \ref{proof_rmatrix}.
\begin{thrm}\label{rmatrix}
The Lax form $W (\lambda)$ is Hamiltonian, with Hamiltonian vector field
\begin{eq}
\label{Lax_vector_field}
\xi_W(\lambda) = \sum_{k=1}^\infty \left( - \parder{Q^{(k)}(\lambda)}{f_1} \partial_{e_k} + \parder{Q^{(k)}(\lambda)}{e_1} \partial_{f_k}\right)\,.
\end{eq}
Its multi-time Poisson bracket possesses the linear Sklyanin bracket structure \ie
\begin{eq}
\label{PBLax}
    \cpb{W_1(\lambda)}{W_2(\mu)} = [r_{12}(\lambda-\mu),{W_1(\lambda)}+{W_2(\mu)}]\,,
\end{eq}
where $r_{12}(\lambda,\mu)$ is the so-called rational classical $r$-matrix given by
\begin{eq}
r_{12}(\lambda) = -\frac{P_{12}}{\lambda}\,.
\end{eq}
\end{thrm}
\begin{remark}
We have already shown directly that our multi-time Poisson bracket $\cpb{~}{~}$ satisfies the Jacobi identity for $0$- and $1$-forms. In the case of $1$-forms, this is also a corollary of Theorem \ref{rmatrix} since $W(\lambda)$ contains all the coordinates of our phase space and it is known that the rational $r$-matrix satisfies the classical Yang-Baxter equation which implies the Jacobi identity.
\end{remark}

\section{Hamiltonian multiform description of the AKNS hierarchy}\label{Section_Hammultiform}

\subsection{Multiform Hamilton equations for the AKNS hierarchy}

According to formula \cite[Definition 2.3]{Caudrelier_Stoppato_2020_2}, the coefficients of the Hamiltonian multiform $\displaystyle \Ham = \sum_{i<j=1}^\infty H_{ij} \, dx^{ij}$ associated to $\Lag$ and $\Omegaone$ are given by
\begin{eq}
H_{ij} = \ip{\tilpartial_i}\omegaone_j - \ip{\tilpartial_j}\omegaone_i - L_{ij}\,.
\end{eq}
As is now customary, we rewrite this in generating form as
\begin{eq}
\Ham(\lambda,\mu) = \ip{{\tilde D}_\lambda} \Omegaone(\mu) - \ip{{\tilde D}_\mu} \Omegaone(\lambda) - \Lag(\lambda,\mu)\,,
\end{eq}
where we introduce the notation $\displaystyle{\tilde D}_\lambda = \sum_{i=0}^\infty \tilpartial_i / \lambda^{i+1}$ in line with \eqref{gen_der}. 
\begin{lmm}\label{Legendre_lmm} 
The following holds
\begin{eq}
\Ham(\lambda,\mu) = V(\lambda,\mu) = - \frac{1}{2} \Tr{ \frac{(Q(\lambda)-Q(\mu))^2}{\lambda-\mu} }\,,
\end{eq}
or, in components,
\begin{eq}
\label{Hpq}
H_{ij} = \Tr \sum_{k=0}^i Q_k Q_{i+j-k+1}\,.
\end{eq}
Hence, $\Ham(\lambda,\mu)$ satisfies the closure relation.
\end{lmm}
\begin{proof}
A direct calculation shows that $\ip{{\tilde D}_\lambda} \Omegaone(\mu) - \ip{{\tilde D}_\mu} \Omegaone(\lambda) = K(\lambda,\mu)$ hence $\Ham(\lambda,\mu) = V(\lambda,\mu)$. Performing a series expansion of the explicit formula for $V(\lambda,\mu)$, one obtains \eqref{Hpq}. Finally, the closure relation of $\Ham$ is a general result \cite[Corollary 2.6]{Caudrelier_Stoppato_2020_2} but here, we get a direct confirmation from the structure of the proof of Theorem \ref{Lag_multi} which established that $V$ is closed on the equations of motion, separately from $K$.
\end{proof}
For completeness, we now check the validity of the general result of Equation \eqref{Ham_eqs} in our case. 
\begin{prp}
The multiform Hamilton equations associated to $\Ham$ and $\Omega$
are equivalent to 
\begin{eq}
D_\lambda Q(\mu)=\frac{[Q(\lambda),Q(\mu)]}{\lambda-\mu}\,.
\end{eq}
\end{prp}
\begin{proof}
The multiform Hamilton equations read
\begin{eq}
\delta \Ham = \sum_j dx^j \wedge \ip{\tilpartial_j} \Omega\,,
\end{eq}
or, in components,
\begin{eq}
\delta H_{ij} = \ip{\tilpartial_j}\omega_i - \ip{\tilpartial_i}\omega_j\,.
\end{eq}
This is reformulated in generating form as,
\begin{eq}
\delta \Ham (\lambda,\mu) = \ip{{\tilde D}_\mu}\Omega(\lambda) - \ip{{\tilde D}_\lambda}\Omega(\mu)\,.
\end{eq}
We have already computed $\delta \Ham (\lambda,\mu) = \delta V(\lambda,\mu) $ as
\begin{eq}
\delta \Ham(\lambda,\mu) =  \Tr \bigg( \frac{1}{\mu-\lambda}\phi(\lambda)^{-1} [Q(\mu),Q(\lambda)] \delta\phi(\lambda)  - \frac{1}{\lambda-\mu}\phi(\mu)^{-1}[Q(\lambda),Q(\mu)]\delta \phi(\mu)\bigg)\,.
\end{eq}
Now the right hand-side is
\begin{eqsplit}
    \ip{{\tilde D}_\mu}\Omega(\lambda) - \ip{{\tilde D}_\lambda}\Omega(\mu) & = \Tr \bigg( - Q_0 \phi(\lambda)^{-1} D_\mu \phi(\lambda) \phi(\lambda)^{-1} \delta \phi(\lambda) + Q_0 \phi(\lambda)^{-1} \delta \phi(\lambda) \phi(\lambda)^{-1} D_\mu \phi(\lambda)\\
   & +Q_0 \phi(\mu)^{-1} D_\lambda \phi(\mu) \phi(\mu)^{-1} \delta \phi(\mu) - Q_0 \phi(\mu)^{-1} \delta \phi(\mu) \phi(\mu)^{-1} D_\lambda \phi(\mu)\bigg)\\
    & = \Tr \left( \phi^{-1}(\lambda) D_\mu Q(\lambda) \delta \phi(\lambda) - \phi^{-1}(\mu) D_\lambda Q(\mu) \delta \phi(\mu) \right)\,.
\end{eqsplit}
The result follows by reading the coefficient of $\delta \phi(\mu)$ or equivalently $\delta \phi(\lambda)$.
\end{proof}

\subsection{Zero curvature equations as multiform Hamilton equations}

It is one of the most important results of the theory of integrable classical field theories that their zero curvature representation admits a Hamiltonian formulation. This is established for all famous models and goes as follows, for the example of the NLS equation, see e.g. \cite{Faddeev_Takhtajan_2007}. We choose $Q^{(1)}(x,\lambda)=-i\lambda\sigma_3+Q_1(x)$ satisfying the Sklyanin linear Poisson algebra \eqref{SklyaninPB}, and $H$ being the (appropriate) Hamiltonian extracted from the monodromy matrix of the spectral problem $(\partial_x-Q^{(1)})\Psi=0$. One can then show that the time flow induced on $Q^{(1)}(x,\lambda)$ by $H$ \ie
\begin{eq}
\partial_tQ^{(1)}(x,\lambda)=\{H,Q^{(1)}(x,\lambda)\}
\end{eq}
takes the form of a zero curvature equation
\begin{eq}
\partial_tQ^{(1)}(x,\lambda)=\{H,Q^{(1)}(x,\lambda)\}=\partial_x Q^{(2)}(x,\lambda)+[Q^{(2)}(x,\lambda),Q^{(1)}(x,\lambda)]\,.
\end{eq}
$Q^{(2)}(x,\lambda)$ is derived from the classical $r$-matrix and the monodromy matrix.
In \cite{Caudrelier_Stoppato_2020} the authors cast this result into a covariant framework, for the NLS and mKdV equations separately: the covariant Hamilton equations for the Lax form associated to each equation (thus containing only the two relevant $Q^{(j)}(x,\lambda)$) produce the respective zero curvature condition. \\

Here, we are in a position to prove the analogous result for the whole AKNS hierarchy at once, thanks to our Hamiltonian multiform and multi-time Poisson bracket. The following is the main result of this section
\begin{thrm}\label{HamiltonianZCE}
The multiform Hamilton equations for the Lax form $\displaystyle W(\lambda) = \sum_{k=0}^\infty Q^{(k)}(\lambda) \, dx^k$, \ie
\begin{eq}
    d W(\lambda) = \sum_{i<j=1}^\infty \cpb{H_{ij}}{W(\lambda)} \, dx^{ij}\,,
\end{eq}
are equivalent to the complete set of zero curvature equations of the AKNS hierarchy\footnote{Recall that we pointed out that in turn, this set of zero curvature equations is equivalent to the set of equations \eqref{FNR_tNequations}.}
\begin{equation}
    \partial_i Q^{(j)}(\lambda) - \partial_j Q^{(i)}(\lambda) = [Q^{(i)}(\lambda),Q^{(j)}(\lambda)]\qquad \forall i<j \,.
\end{equation}
\end{thrm}
The proof is given in Appendix \ref{proof_HamiltonianZCE}.

\subsection{Conservation laws}\label{conservation}

An important aspect of the theory of integrable PDEs is the presence of an infinite number of conservation laws. When formulated in the traditional Hamiltonian framework, this means that the Hamiltonian $h$ is one element in a countable family of independent Hamiltonians functions $h_k$, $k\ge 1$ say, which are in involution with respect to the Poisson bracket $\pb{~}{~}$ used to write the PDE of interest as a Hamiltonian system
\begin{eq}
\label{involution}
\pb{h_n}{h_m}=0\,,~~n,m\ge 1\,.
\end{eq}
When using the classical $r$-matrix formalism to tackle this problem, the strategy is to compute the Poisson bracket of the entries of the monodromy matrix, starting from the linear Sklyanin bracket and deduce that \eqref{involution} holds, where the $h_k$'s are extracted in an appropriate way from the monodromy matrix. The details depend on the model and on the boundary conditions imposed on the fields, see e.g. \cite{Faddeev_Takhtajan_2007} for a full account of this procedure. 

It is natural to ask what happens to this approach in our context and where to find the standard Hamiltonians $h_n$, which are rather different from the coefficients $H_{ij}$ of our Hamiltonian multiform. In our opinion, it is rather remarkable to we do not need to resort to any notion of monodromy matrix to answer this question. Instead, all the required information is encoded in our Hamiltonian multiform $\Ham$ and the notion of conservation law which we naturally define as follows.
\begin{dfn}
A conservation law is a Hamiltonian $1$-form $\displaystyle A=\sum_{i=0}^\infty A_i \, dx^i$ such that $dA =0$ on the equations of motion.
\end{dfn}
In components, this yields of course the familiar form of an infinite family of conservation laws
\begin{eq}
\partial_i A_j - \partial_j A_i=0\,.
\end{eq}
We have proved in \cite{Caudrelier_Stoppato_2020_2} that $A$ is a conservation law if it Poisson-commutes with the Hamiltonian multiform, \ie
\begin{eq}
\sum_{i<j}^\infty\cpb{H_{ij}}{A}\,dx^{ij}=0\,.
\end{eq}
The familiar conservation laws and related Hamiltonians $h_k$ are obtained by considering the following $1$-form.
\begin{prp}[Conservation laws]
The form 
\begin{eq}
A= \sum_{k=0}^\infty A_k\,dx^k\,,~~A_k=a_{k+1}
\end{eq}
is a conservation law.
\end{prp}
\begin{proof}
From \eqref{a_ef}, we find $\displaystyle A_k = a_{k+1} = \sum_{i=1}^k e_i f_{k+1-i}$ so that $\parder{A_i}{f_j}=e_{i+1-j} = \parder{A_{i+1}}{f_{j+1}}$  and $\parder{A_i}{e_j}=f_{i+1-j} = \parder{A_{i+1}}{e_{j+1}}$. Hence $A$ is Hamiltonian. Now,
\begin{equation}
dA = \ip{\xi_A}\delta \Ham=  \sum_{m<n}  \sum_{k=1}^\infty \left( - \parder{A_k}{f_1} \parder{H_{mn}}{e_k} + \parder{A_k}{e_1} \parder{H_{mn}}{f_k} \right)\, dx^{mn} =     \sum_{m<n} \sum_{k=1}^n \left(-e_k \parder{H_{mn}}{e_k} +f_k \parder{H_{mn}}{f_k}\right)\, dx^{mn}
\end{equation}
where we have used $\parder{A_k}{e_1} = f_k$, $\parder{A_k}{f_1} = e_k$, and the fact that $\parder{H_{mn}}{e_k}=\parder{H_{mn}}{f_k}=0$ if $k>n$ (without loss of generality, we consider $m<n$).
From the explicit expression of $\Ham(z,w)$, a direct but tedious argument shows that each $H_{mn}$ is in fact a polynomial in $e_1,\dots,e_n,f_1,\dots,f_n$ of the form 
\begin{equation}
    H_{mn} = \sum_{(i),(j) \in \N^n}h_{(i)(j)} (e)^{(i)} (f)^{(j)}\,,
\end{equation}
where the sum is finite (only a finite number of coefficients $h_{(i)(j)}\in\C$ are non zero) and we have used the notations $(e)^{(i)} = e_1^{i_1}e_2^{i_2}\dots e_m^{i_m}$, $(f)^{(j)}= f_1^{j_1}f_2^{j_2}\dots f_n^{j_n}$, and has the property that 
$\displaystyle \sum_{k=1}^n i_k = \sum_{k=1}^n j_k$. The result then follows since $\displaystyle\sum_{k=1}^n e_k \parder{}{e_k}$ and $\displaystyle\sum_{k=1}^n f_k \parder{}{f_k}$ are Euler operators with respect to the coordinates $e_k$ and $f_k$ respectively.
\end{proof}
This result provides a reinterpretation of the known fact the quantities $h_k=\frac{1}{k}\int a_{k+1}\,dx^1$, viewed as the traditional hierarchy of standard, single-time, Hamiltonians are indeed constant of the motion and in involution with respect to the traditional (single-time) Poisson bracket $\pb{~}{~}_1$ (see e.g. \cite[Section 9.3]{Dickey_2003}).

\section{Recovering previous results and the first three times}\label{Section_3times}

It is straightforward to recover our previous results \cite{Caudrelier_Stoppato_2020} by ``freezing'' all times except a given pair. This singles out a single $1+1$-dimensional field theory within the hierarchy and our Lagrangian multiform, symplectic multiform, Hamiltonian multiform and multi-time Poisson bracket reduce respectively to a Lagrangian, multisymplectic form, covariant Hamiltonian and covariant Poisson bracket. 

As the simplest example, let use freeze all times except $x^1=x$ and $x^2=t$: we specialise to NLS and recover all the results of Section 4.2 in \cite{Caudrelier_Stoppato_2020} by direct calculation. The Lax form is simply
\begin{eq}
 W(\lambda)=Q^{(1)}(\lambda)\,dx+Q^{(2)}(\lambda)\,dt\,,
\end{eq}
which can be computed using again the coordinates $q$, $r$ and derivatives with respect to $x$ for instance to reproduce the well known NLS Lax pair. The Lagrangian multiform reduces to $\cL=L_{12}dx\wedge dt$ where $L_{12}$ is given in \eqref{NLS_Lag} while the Hamiltonian multiform only involves $H_{12}$. Using our general formula, 
\begin{eq}
H_{ij} = \sum_{k=0}^i (2 a_k a_{i+j+1-k} + b_k c_{i+j+1-k} + c_k b_{i+j+1-k} )\,,
\end{eq}
we find
\begin{eqsplit}
H_{12} = &2a_0a_4 + b_0c_4 + c_0 b_4 + 2a_1 a_3 + b_1 c_3 + c_1 b_3 \\
=& -2i(e_1 f_3 + e_2 f_2 + e_3 f_1) + 2ie_1(f_3 + \frac{i}{4}e_1 f_1^2) + 2if_1(e_3 + \frac{i}{4} e_1^2 f_1)\\
=& -2i e_2 f_2 - e_1^2 f_1^2\\
=& -\frac{1}{4}(q_1r_1 - q^2 r^2)\,.
\end{eqsplit}
This is the covariant Hamiltonian for NLS found in \cite{Caudrelier_Stoppato_2020} (up to an irrelevant factor). The symplectic multiform boils down to the following multisymplectic form 
\begin{subgather}
\Omega=\omega_1\wedge dx    +\omega_2\wedge dt\,,\\
\omega_1  = \frac{i}{2} \delta q \wedge \delta r\,,~~
    \omega_2 = \frac{1}{4} \delta r \wedge \delta q_1 + \frac{1}{4}\delta q \wedge \delta r_1\,,
\end{subgather}
also found first in \cite{Caudrelier_Stoppato_2020} (up to irrelevant factors). It gives rise to a covariant Poisson bracket which is simply the reduction of our multi-time Poisson bracket to only two times and our main results, Theorems \ref{rmatrix} and \ref{HamiltonianZCE} restrict accordingly to the results of \cite{Caudrelier_Stoppato_2020}.

We stress however that we can instead choose any pair of times $x^n$ and $x^k$ and apply the same reasoning. Doing so provides a way to unify the results in \cite{Avan_Caudrelier_2017} which established the $r$-matrix structure of dual Lax pairs for an arbitrary pair of times and the results in \cite{Caudrelier_Stoppato_2020} which provided a covariant formulation of this structure but only for the pair of times $(x^1,x^2)$ and $(x^1,x^3)$. 

The salient features of the multiform theory appear when at least three times are combined together. In general, the coefficients $L_{1n}$ (resp. $H_{1n}$) are not too difficult to construct but all the other ones are, and indeed up to now, it was not known how to obtain them. For instance, freezing all times except $x^1$, $x^2$, $x^3$, the coefficient $L_{23}$ was first obtained in \cite{Sleigh_Nijhoff_Caudrelier_2019} by complicated calculations. Here, we obtain it rather easily, see \eqref{eq_L23}, as well as the associated coefficient $H_{23}$ in the Hamiltonian multiform which reads
\begin{eqnarray}
       H_{23} &=&- 2i e_3 f_3 + \frac{1}{2} e_1 f_1(f_1 e_3 + e_1 f_3) - (e_1 f_2 + f_1 e_2)^2 + \frac{i}{8}e_1^3 f_1^3\\
        &=&- \frac{1}{16}q_{11} r_{11} + \frac{qr}{8}(rq_{11} + qr_{11}) - \frac{1}{16}(r q_1 - q r_1)^2 - \frac{1}{4} q^3 r^3
\end{eqnarray}
For completeness, let us also give 
\begin{gather}
     H_{13} = -2i( e_2 f_3 + e_3 f_2) - \frac{3}{2} e_1 f_1 (f_1 e_2 + e_1 f_2)\,,\\
     =\frac{i}{8}\left(q_1r_{11} - r_1q_{11} \right)\,.
\end{gather}
We remark that these coefficients differ from those in \cite{Caudrelier_Stoppato_2020_2} by a factor $\frac{i}{2}$.
In the rest of this section, we illustrate in every detail the calculations involved in our general results when restricted to the first three times. This has only pedagogical value. We hope that this will help the reader familiarise themselves with some of the new formalism while dealing with the most familiar and easiest levels of the AKNS hierarchy.
We now turn to the symplectic multiform $\Omega = \omega_1\wedge dx^1 + \omega_2 \wedge dx^2 + \omega_3 \wedge dx^3$, where
\begin{eq}
\omega_1 = \delta f_1\wedge \delta e_1 \,, \quad \omega_2 = \delta f_1 \wedge \delta e_2 + \delta f_2 \wedge \delta e_1\,, \quad \omega_3 = \delta f_1 \wedge \delta e_3 + \delta f_2 \wedge \delta e_2 + \delta f_3 \wedge \delta e_1\,.
\end{eq}
As done above for $\omega_1$ and $\omega_2$, it is interesting to write $\omega_3$ using $b_1 = q$, $c_1 = r$ and their derivatives with respect to $x^1$, denoted by $q_1$, $r_1$, $q_{11}$, $r_{11}$. We find
\begin{subgather}
    \omega_3 = \frac{i}{8} \delta r \wedge \delta q_{11} + \frac{i}{8} \delta r_{11} \wedge \delta q  + \frac{i}{8} \delta q_1 \wedge \delta r_1 + \frac{3iqr}{4} \delta q \wedge \delta r\,,
\end{subgather}
and we remark that they also differ from the ones in \cite{Caudrelier_Stoppato_2020_2} by the same factor $\frac{i}{2}$, so that the multiform Hamilton equations $\delta \Ham = \sum_j dx^j \wedge \ip{\tilpartial_j} \Omega$ are the same. Let us compute them, in the new $e$ and $f$ coordinates. In components we have
\begin{itemize}
    \item $\delta H_{12} = \ip{\partial_2}{\omega_1} - \ip{\partial_1}{\omega_2}$:
    \begin{subgather}
        \partial_1 f_1 = 2i f_2\,, \qquad \partial_1 e_1= - 2i e_2\,,\\
  \partial_1 f_2 - \partial_2 f_1 = 2 e_1 f_1^2\,, \qquad \partial_2 e_1 - \partial_1 e_2 = 2 e_1^2 f_1\,.
    \end{subgather}
    The top equations give the relations $b_2 = \frac{i}{2} \partial_1 b_1 = \frac{i}{2} q_1$ and $c_2 = - \frac{i}{2} \partial_1   c_1 = -\frac{i}{2} r_1$, and the bottom ones give the NLS equations.
    \item $\delta H_{13} = \ip{\partial_3}{\omega_1} - \ip{\partial_1}{\omega_3}$:
    \begin{subgather}
  \partial_1 f_1 = 2i f_2\,, \qquad \partial_1 e_1= - 2i e_2\,,\\
  \partial_1 f_2 = 2i f_3 + \frac{3}{2} e_1 f_1^2\,, \qquad \partial_1 e_2 = - 2i e_3 - \frac{3}{2} e_1^2 f_1\,,\\
  \partial_1 f_3 - \partial_3 f_1 = \frac{3}{2} e_2 f_1^2 + 3 e_1 f_1 f_2\,, \qquad \partial_3 e_1 - \partial_1 e_3 = \frac{3}{2} e_1^2 f_2 + 3 e_1 f_1 e_2\, ,
    \end{subgather}
    where the top four equations give the relations  $b_2 = \frac{i}{2} q_1$ and $c_2 =  -\frac{i}{2} r_1$, and $b_3= -\frac{1}{4}q_{11} + \frac{1}{2}q^2r$ and $c_3 = -\frac{1}{4}r_{11} + \frac{1}{2}q r^2$, and the bottom ones are the mKdV equations.
    \item  $\delta H_{23} = \ip{\partial_3}{\omega_2} - \ip{\partial_2}{\omega_3}$:
    \begin{subgather}
 \partial_2 f_1 = 2i f_3 - \frac{1}{2} e_1 f_1^2\,, \qquad \partial_2 e_1 = -2 i e_3 + \frac{1}{2} e_1^2 f_1\,,\\
 \partial_2 f_2 - \partial_3 f_1 = 2 f_1^2 e_2 + 2 e_1 f_1 f_2\,, \qquad \partial_3 e_1 - \partial_2 e_2 = 2 e_1^2 f_2 + 2e_1 f_1 e_2\,,\\
 \partial_3 f_2 - \partial_2 f_3 = \frac{1}{2} f_1^2 e_3 + e_1 f_1 f_3 + \frac{3i}{8} e_1^2 f_1^3 - 2 e_1 f_2^2 - 2 f_1 e_2 f_2\,,\\
 \partial_2 e_3 - \partial_3 e_2 = \frac{1}{2} e_1^2 f_3 + e_1 f_1 e_3 + \frac{3i}{8} e_1^3 f_1^2 - 2 f_1 e_2^2 - 2 e_1 e_2 f_2\,,
\end{subgather}
which reduce to differential consequences of the previous equations.
\end{itemize}
The single-time Poisson brackets $\pb{\quad}{\quad}_k$ for $k=1,2,3$
\begin{eq}
\pb{\quad}{\quad}_k = \sum_{i=1}^k \left(\parder{}{f_i}\parder{}{e_{k+1-i}} - \parder{}{e_{k+1-i}}\parder{}{f_i} \right)\,,
\end{eq}
can be re-expressed in the $q$ and $r$ coordinates as
\begin{eq}
\pb{\quad}{\quad}_1 = 2i \left( \parder{}{r} \parder{}{q} - \parder{}{q} \parder{}{r} \right)\,,
\end{eq}
\begin{eq}
\pb{\quad}{\quad}_2 = 4 \left( \parder{}{r} \parder{}{q_1} + \parder{}{q} \parder{}{r_1} - \parder{}{q_1} \parder{}{r} - \parder{}{r_1} \parder{}{q}\right)\,,
\end{eq}
\begin{eqsplit}
\pb{\quad}{\quad}_3 =-8i\Big(&  \parder{}{r}\parder{}{q_{11}} + \parder{}{r_{11}}\parder{}{q} + \parder{}{q_1}\parder{}{r_1} + 6qr \parder{}{q_{11}} \parder{}{r_{11}}\\
&- \parder{}{q_{11}}\parder{}{r} - \parder{}{q}\parder{}{r_{11}} - \parder{}{r_1}\parder{}{q_1} - 6qr \parder{}{r_{11}} \parder{}{q_{11}}\Big)\,.
\end{eqsplit}
These differ from the ones obtained in \cite{Caudrelier_Stoppato_2020_2} by a factor $-2i$.\\

We will now show how to obtain the classical $r$-matrix structure within the multi-time Poisson brackets for the first three times. We will use the first three Lax matrices repackaged into the Lax form $W(\lambda) = Q^{(1)}(\lambda)\,dx^1 + Q^{(2)}(\lambda)\, dx^2 + Q^{(3)}(\lambda)\,dx^3$
\begin{subgather}
    W^+(\lambda) = b_1 \,dx^1 + (\lambda b_1 +b_2) \, dx^2 +(\lambda^2 b_1 + \lambda b_2 + b_3) \, dx^3\,,\\
    W^-(\lambda) = c_1 \,dx^1 + (\lambda c_1 +c_2) \, dx^2 +(\lambda^2 c_1 + \lambda c_2 + c_3) \, dx^3\,,\\
    W^3(\lambda) = -i \lambda dx^1 +( -i\lambda^2 - \frac{i}{2} b_1 c_1)\, dx^1 + (-i \lambda^3 - \frac{i \lambda}{2}b_1 c_1 - \frac{i}{2} (b_1 c_2 + c_1 b_2))\, dx^3 \,,
\end{subgather}
which we can also write in terms of the coordinates $e$ and $f$ as in
\begin{subgather}
    W^+(\lambda) = \sqrt{2i}e_1 \,dx^1 +\sqrt{2i} (\lambda e_1 +e_2) \, dx^2 +\sqrt{2i}(\lambda^2 e_1 + \lambda e_2 + e_3 + \frac{i}{4} e_1^2 f_1) \, dx^3\,,\\
    W^-(\lambda) =\sqrt{2i} f_1 \,dx^1 + \sqrt{2i}(\lambda f_1 +f_2) \, dx^2 +\sqrt{2i}(\lambda^2 f_1 + \lambda f_2 + f_3 + \frac{i}{4} e_1 f_1^2) \, dx^3\,,\\
    W^3(\lambda) = -i \lambda dx^1 +( -i\lambda^2 +e_1f_1)\, dx^1 + (-i \lambda^3 +\lambda e_1 f_1 +e_1 f_2 + e_2 f_1)\, dx^3 \,.
\end{subgather}
We can then compute the Hamiltonian vector field associated to each component of the Lax form:
\begin{subgather}
     \xi_{W^+}(\lambda) = \sqrt{2i}\left( \partial_{f_1} + \lambda \partial_{f_2} + (\lambda^2 + \frac{i}{2}e_1f_1)\partial_{f_3} - \frac{i}{4} e_1^2 \partial_{e_3} \right)\,,\\
      \xi_{W^-}(\lambda) = \sqrt{2i}\left( -\partial_{e_1} - \lambda \partial_{e_2} + (-\lambda^2 - \frac{i}{2}e_1f_1)\partial_{e_3} + \frac{i}{4} f_1^2 \partial_{f_3} \right)\,,\\
      \xi_{W^3}(\lambda) = - e_1 \partial_{e_2} + f_1 \partial_{f_2} + (-\lambda e_1 - e_2)\partial_{e_3} + (\lambda f_1 + f_2) \partial_{f_3}\,.
\end{subgather}
Let us now compute the multi-time Poisson bracket.
\begin{eqsplit}
    \cpb{W_1(\lambda)}{W_2(\mu)} =&\sum_{i,j=+,-,3} \cpb{W^i(\lambda)}{W^j(\mu)} \sigma_i \otimes \sigma_j\\
    =&\cpb{W^+ (\lambda)}{W^+ (\mu)} \sigma_+ \otimes \sigma_+ + \cpb{W^+ (\lambda)}{W^- (\mu)} \sigma_+ \otimes \sigma_-\\
    & + \cpb{W^+ (\lambda)}{W^3 (\mu)} \sigma_+ \otimes \sigma_3 + \cpb{W^- (\lambda)}{W^+ (\mu)} \sigma_- \otimes \sigma_+\\
    &+ \cpb{W^- (\lambda)}{W^- (\mu)} \sigma_- \otimes \sigma_- + \cpb{W^- (\lambda)}{W^3 (\mu)} \sigma_- \otimes \sigma_3\\
    &+ \cpb{W^3 (\lambda)}{W^+ (\mu)} \sigma_3 \otimes \sigma_+ + \cpb{W^3 (\lambda)}{W^- (\mu)} \sigma_3 \otimes \sigma_-\\
    &+ \cpb{W^3 (\lambda)}{W^3 (\mu)} \sigma_3 \otimes \sigma_3\,.
\end{eqsplit}
The reader can check that $\cpb{W^+ (\lambda)}{W^+ (\mu)} = \cpb{W^- (\lambda)}{W^- (\mu)}= \cpb{W^3 (\lambda)}{W^3 (\mu)}=0$, while the other non-zero Poisson brackets are
\begin{subgather}
     \cpb{W^+(\lambda)}{W^-(\mu)} = -2i\, dx^1 - 2i(\lambda + \mu)\, dx^2 - 2i \, (\lambda^2 + \lambda \mu + \mu^2 + i e_1 f_1)\, dx^3\,\\
     \cpb{W^-(\lambda)}{W^+(\mu)} = 2i\, dx^1 + 2i(\lambda + \mu)\, dx^2 + 2i \, (\lambda^2 + \lambda \mu + \mu^2 + i e_1 f_1)\, dx^3\,\\
     \cpb{W^+(\lambda)}{W^3(\mu)} = - \sqrt{2i} e_1 \, dx^2 - \sqrt{2i}( (\lambda + \mu) e_1 + e_2)\, dx^3\,,\\
     \cpb{W^3(\lambda)}{W^+(\mu)} =  \sqrt{2i} e_1 \, dx^2 + \sqrt{2i}( (\lambda + \mu) e_1 + e_2)\, dx^3\,,\\
     \cpb{W^-(\lambda)}{W^3(\mu)} =  \sqrt{2i} f_1 \, dx^2 + \sqrt{2i}( (\lambda + \mu) f_1 + f_2)\, dx^3\,,\\
      \cpb{W^3(\lambda)}{W^-(\mu)} = - \sqrt{2i} f_1 \, dx^2 - \sqrt{2i}( (\lambda + \mu) f_1 + f_2)\, dx^3\,.
\end{subgather}
Adding everything together one realises that $\cpb{W_1(\lambda)}{W_2(\mu)}=[\frac{P_{12}}{\mu - \lambda}, W_1(\lambda) + W_2(\mu)]$, as desired.\\

Let us verify that for the first three times
\begin{eq}
 \sum_{i<j}\cpb{H_{ij}}{W(\lambda)} =W(\lambda) \wedge W(\lambda)= \sum_{i<j}[Q^{(i)}(\lambda),Q^{(j)}(\lambda)] \, dx^{ij}
\end{eq}
or, in components,
\begin{eq}
 \cpb{H_{ij}}{W(\lambda)} = [Q^{(i)}(\lambda),Q^{(j)}(\lambda)]\,.
\end{eq}
We write explicitly the $(1,2)$ term. The coefficient of the Hamiltonian multiform $H_{12} = - 2i e_2 f_2 - e_1^2 f_1^2$ has Hamiltonian vector field
\begin{eq}
  \xi_{12} = 2 e_1^2 f_1 \partial_{e_1} \wedge \partial_1 - 2 e_1 f_1^2 \partial_{f_1} \wedge \partial_1 + 2i e_2 \partial_{e_1} \wedge \partial_2 - 2i f_2 \partial_{f_1} \wedge \partial_2\,,
\end{eq}
so that the left hand-side reads
\begin{eqsplit}
     \cpb{H_{12}}{W(\lambda)} =& \ip{\xi_{12}}{\delta W(\lambda)}\\
    =&\ip{\xi_{12}}{(e_1 \delta f_1\wedge dx^2 + f_1 \delta e_1 \wedge dx^2)} \sigma_3 \\
    &+\ip{\xi_{12}}{( \sqrt{2i} \delta e_1 \wedge dx^1 + \sqrt{2i} \delta e_2 \wedge dx^2 + \sqrt{2i} \lambda \delta e_1 \wedge dx^2)} \sigma_+\\
    &+\ip{\xi_{12}}{( \sqrt{2i} \delta f_1 \wedge dx^1 + \sqrt{2i} \delta f_2 \wedge dx^2 + \sqrt{2i} \lambda \delta f_1 \wedge dx^2)} \sigma_-\\
    =& 2i( e_1 f_2 - f_1 e_2) \sigma_3 + \sqrt{2i}(-2 e_1^2 f_1 - 2i \lambda e_2) \sigma_+ + \sqrt{2i}(2e_1 f_1^2 + 2i \lambda f_2) \sigma_- \\
    =& [Q^{(1)}(\lambda), Q^{(2)}(\lambda)]\,.
\end{eqsplit}
Similarly one obtains $\cpb{H_{13}}{W(\lambda)} = [Q^{(1)}(\lambda),Q^{(3)}(\lambda)]$ and $\cpb{H_{23}}{W(\lambda)} = [Q^{(2)}(\lambda),Q^{(3)}(\lambda)]$.\\

We can also verify that $A = a_2 dx^1 + a_3 dx^2 + a_4 dx^3$ is indeed a conservation law in the usual coordinates $q$ and $r$. In fact, we have that
\begin{subgather}
     a_2 = e_1 f_1 =-\frac{i}{2}qr\,,\\
     a_3 = e_1 f_2 + e_2 f_1 = \frac{1}{4} (q_1 r - q r_1)\,,\\
     a_4 = e_1 f_3 + e_2 f_2 + e_3 f_1 = \frac{i}{8}q r_{11} + \frac{i}{8} q_{11} r - \frac{3i}{8}q^2 r^2 - \frac{i}{8}q_1 r_1\,.
\end{subgather}
Imposing $dA=0$ is equivalent to the equations
\begin{subgather}
     \partial_1 a_3 = \partial_2 a_2\,,\\
     \partial_1 a_4 = \partial_3 a_2\,,\\
     \partial_2 a_4 = \partial_3 a_3\,,
\end{subgather}
which hold on the equations of motion.

\section{Conclusions and perspectives}\label{Conclusions section}

This work constitutes progress towards the understanding of the role played by multi-time consistency and the application of covariant Hamiltonian field theory to integrable systems, which is a line of research that started with \cite{Caudrelier_Stoppato_2020} and continued with \cite{Caudrelier_Stoppato_2020_2} and \cite{Caudrelier_Stoppato_Vicedo_2020}. In this paper, we have presented a new Lagrangian and Hamiltonian multiforms to describe the complete Ablowitz-Kaup-Newell-Segur (AKNS) hierarchy, and recognised in this description the main features of integrability, including the presence of an infinite number of conservation laws and the classical $r$-matrix structure. Unlike the usual approach to an integrable hierarchy, our formalism preserves equal footing with respect to all the times of the integrable hierarchy, and avoids altogether the infinite-dimensional Hamiltonian formalism in favour of a more \enquote{finite-dimensional} approach to a field theory. The Lagrangian multiform is written as a double generating series in the negative loop algebra, and is proved to satisfy the closure relation. In turn, via a \enquote{multi-time Legendre transform}, we obtain a Hamiltonian and a symplectic multiform, which encapsulate respectively all the covariant Hamiltonians and the symplectic structures of the hierarchy. The classical $r$-matrix structure is found within a multi-time Poisson bracket (that is derived from the symplectic multiform), which contains all the single-time Poisson brackets of the hierarchy. The set of zero curvature equations are obtained as multiform Hamilton equations for the complete Lax connection. Conservation laws have a vanishing multi-time Poisson bracket with the Hamiltonian multiform, emulating the familiar condition in finite dimensional Hamiltonian mechanics. The well-known fluxes and conserved quantities are re-obtained naturally from this requirement.\\

These results point to some interesting open questions. In \cite{Caudrelier_Stoppato_Vicedo_2020} we have identified a covariant equivalent of the famous relation \enquote{$ H = \Tr L^2$} between the Hamiltonian multiform and the Lax connection. An open question is if there is such a formula for the AKNS hierarchy as well, and what are its implications. The conservation laws are obtained in this paper without resorting to the monodromy matrix \cite{Sklyanin_1982}, and without involving the $r$-matrix structure at the group level, which is the starting point of the traditional and well-known (quantum) Inverse Scattering Method. This is in our opinion certainly remarkable, and it needs to be understood whether the monodromy matrix is really out of the picture and if so, what is the reason. Moreover, we were not able to relate the closure of the Lagrangian multiform with the \enquote{mutual involution} of the single-time Hamiltonians, which would connect our work to the one of \cite{Vermeeren_2020}. Finally, we believe that a similar approach will be useful in describing other integrable hierarchies, such as the ones containing the (potential) Korteweg-de Vries equation or the sine-Gordon equation, which were first analysed with a Hamiltonian multiform formalism in \cite{Caudrelier_Stoppato_2020_2}. This is left for future investigation.\\

As mentioned above, the classical $r$-matrix structure of a single-time Poisson bracket is the starting point of a well-established procedure of canonical quantisation \cite{Sklyanin_1979}. We wish to remark that this work (together with \cite{Caudrelier_Stoppato_2020}, \cite{Caudrelier_Stoppato_2020_2} and partly \cite{Caudrelier_Stoppato_Vicedo_2020}) belongs to a programme whose overarching goal is a new approach to \emph{canonical covariant quantisation} of an integrable system, and builds an important step towards this objective. We believe that the classical $r$-matrix structure within the covariant (and multi-time) Poisson bracket can provide a new outlook on how to perform this canonical quantisation in a covariant fashion.\\

\appendix

\section{The $e$ and $f$ coordinates}\label{efcoordinates_section}
In this section we discuss some of the properties of the coordinates $e$ and $f$, which are defined as
\begin{eq}
e(\lambda) = \frac{b(\lambda)}{\sqrt{i-a(\lambda)}} \,, \qquad f(\lambda) = \frac{c(\lambda)}{\sqrt{i-a(\lambda)}}\,.
\end{eq}
We remember that we are restricting to the subset where $a^2(\lambda) + b(\lambda) c(\lambda) = -1$, which means that $a(\lambda) = e(\lambda)f(\lambda) -i$, in fact
\begin{eq}
\label{a_ef}
e(\lambda) f(\lambda) = \frac{b(\lambda) c(\lambda)}{i-a(\lambda)} = \frac{-1-a^2(\lambda)}{i-a(\lambda)} = i+a(\lambda)\,.
\end{eq}
We list the first few coefficients of $e$ and $f$:
\begin{subgather}
e_0 = 0\,,\qquad f_0=0\,,\\
 e_1 = \frac{1}{\sqrt{2i}}b_1 \,, \qquad f_1 = \frac{1}{\sqrt{2i}}c_1\,,\\
  e_2 = \frac{1}{\sqrt{2i}}b_2\,, \qquad  f_2 = \frac{1}{\sqrt{2i}}c_2\,,\\
   e_3 = \frac{1}{\sqrt{2i}}(b_3 - \frac{1}{8} b_1^2 c_1 )\,, \qquad f_3 = \frac{1}{\sqrt{2i}}(c_3 - \frac{1}{8} b_1 c_1^2 )\,,\\
     e_4 = \frac{1}{\sqrt{2i}}( b_4- \frac{1}{4} b_1 c_1 b_2 - \frac{1}{8} b_1^2 c_2)\,, \qquad f_4 = \frac{1}{\sqrt{2i}}( c_4- \frac{1}{4} b_1 c_1 c_2 - \frac{1}{8} c_1^2 b_2)\,.    
\end{subgather}
Conversely, we have that
\begin{subgather}
     b_1 = \sqrt{2i} e_1\,, \qquad c_1 = \sqrt{2i} f_1\,,\\
     b_2= \sqrt{2i} e_2\,,\qquad  c_2 = \sqrt{2i} f_2\,,\\
     b_3 = \sqrt{2i} ( e_3 + \frac{i}{4}e_1^2 f_1)\,, \qquad c_3 = \sqrt{2i} (f_3 + \frac{i}{4} e_1 f_1^2)\,,\\
     b_4 = \sqrt{2i} ( e_4 + \frac{i}{2}e_1 f_1 e_2 + \frac{i}{4} e_1^2 f_2)\,, \qquad c_4 = \sqrt{2i} (f_4 + \frac{i}{2} e_1 f_1 f_2 + \frac{i}{4} f_1^2 e_2)\,.
\end{subgather}
Also $a = ef -i$, so $a_k = \sum_{i=1}^{k-1} e_i f_{k-i}$:
\begin{equation}
    a_0 = -i\,, \qquad a_1 = 0\,, \qquad a_2 = e_1 f_1\,, \qquad a_3 = e_1 f_2 + e_2 f_1\,, \qquad a_4 = e_1 f_3+ e_2 f_2 + f_1 e_3\,.
\end{equation}
It is also useful to express these relations in terms of the usual $q$ and $r$ coordinates (and their derivatives with respect to $x^1=x$) we have the following identities
\begin{subgather}
    b_1=q\,,\qquad c_1=r\,,\\
    b_2 = \frac{i}{2}q_1\,,\qquad c_2 =- \frac{i}{2}r_1\,,\\
    b_3 = -\frac{1}{4}q_{11} + \frac{1}{2}q^2r\,,\qquad c_3 = -\frac{1}{4}r_{11} + \frac{1}{2}q r^2\,,\\
    b_4 = - \frac{i}{8} q_{111} + \frac{3i}{4}qrq_1\,,\qquad c_4 = \frac{i}{8} r_{111} - \frac{3i}{4} qr r_1\,.
\end{subgather}
\begin{subgather}
    e_1 = \frac{1}{\sqrt{2i}}q\,, \qquad f_1  =  \frac{1}{\sqrt{2i}}r\,,\\
    e_2 =  \frac{1}{\sqrt{2i}}\frac{i}{2}q_1 \,,\qquad f_2 =  -\frac{1}{\sqrt{2i}}\frac{i}{2}r_1\,,\\
    e_3 =  \frac{1}{\sqrt{2i}}\left(-\frac{1}{4}q_{11} + \frac{3}{8}q^2 r \right)\,,\qquad f_3 =  \frac{1}{\sqrt{2i}}\left( -\frac{1}{4} r_{11} + \frac{3}{8}qr^2\right)\,,\\
    e_4 =  \frac{1}{\sqrt{2i}}\left( - \frac{i}{8} q_{111} + \frac{5i}{8} qrq_1 + \frac{i}{16} q^2 r_1 \right)\,,\qquad f_4 =  \frac{1}{\sqrt{2i}}\left( \frac{i}{8}r_{111} - \frac{5i}{8}qrr_1 - \frac{i}{16} q_1r^2 \right)\,.
\end{subgather}
Conversely:
\begin{subgather}
    q = \sqrt{2i} e_1 \,,\qquad r = \sqrt{2i}f_1\,,\\
    q_1 = -\sqrt{2i}2i e_2 \,,\qquad r_1 = \sqrt{2i}2i f_2\,,\\
    q_{11} = \sqrt{2i}\left( -4e_3 + 3i e_1^2 f_1\right)\,,\qquad r_{11} = \sqrt{2i}\left( -4 f_3 + 3i e_1 f_1^2 \right)\,,\\
    q_{111} =\sqrt{2i}\left( 8ie_4 + 20e_1 f_1 e_2- 2e_1^2 f_2\right)\,,\quad r_{111} = \sqrt{2i}\left( -8i f_4 -20 e_1 f_1 f_2 + 2 f_1^2 e_2\right)\,.
\end{subgather}
We can also write the expressions for the derivatives of $Q$ with respect to the coordinates $e$ and $f$:
\begin{equation}
 \frac{\partial Q(\lambda)}{\partial e_k}=\frac{\lambda^{-k}}{\sqrt{i-a(\lambda)}}\begin{pmatrix}
    c(\lambda) & \frac{i-3a(\lambda)}{2}\\
    -\frac{c^2(\lambda)}{2(i-a(\lambda))} & -c(\lambda)
    \end{pmatrix}\,,~~    \frac{\partial Q(\lambda)}{\partial f_k}=\frac{\lambda^{-k}}{\sqrt{i-a(\lambda)}} \begin{pmatrix}
    b(\lambda) & -\frac{b^2(\lambda)}{2(i-a(\lambda))}\\
    \frac{i-3a(\lambda)}{2} & -b(\lambda)
    \end{pmatrix}\,.
\end{equation}
Therefore we have
\begin{subgather}
\label{formula_der_e1}
    \parder{a_i}{e_j} = f_{i-j}\,,\\
    \parder{b_i}{e_j} = \left(\frac{i-3a(\lambda)}{2\sqrt{i-a(\lambda)}}\right)_{i-j} \,,\\
    \label{formula_der_e3}
    \parder{c_i}{e_j} = \left(\frac{-f^2(\lambda)}{2\sqrt{i-a(\lambda)}}\right)_{i-j} = \left(\frac{-c^2(\lambda)}{2(i-a(\lambda))^{3/2}} \right)_{i-j}\,,
\end{subgather}
\begin{subgather}
\label{formula_der_f}
    \parder{a_i}{f_j} = e_{i-j}\,,\\
    \parder{b_i}{f_j} = \left(\frac{-e^2(\lambda)}{2\sqrt{i-a(\lambda)}}\right)_{i-j} = \left(\frac{-b^2(\lambda)}{2(i-a(\lambda))^{3/2}} \right)_{i-j}\,,\\
    \parder{c_i}{f_j} = \left(\frac{i-3a(\lambda)}{2\sqrt{i-a(\lambda)}}\right)_{i-j}\,.
\end{subgather}

\section{Proofs}\label{proof section}
\subsection{Proof of Theorem \ref{Lag_multi}}\label{Proof_Lag_multi}
\begin{proof}
We need to calculate $\delta d K$ and then $\delta d V$. 
We do so with the help of the generating functions as follows. Note that 
\begin{eq}
 dK=\sum_{i<j<k}\left(\partial_iK_{jk}+\partial_kK_{ij}+\partial_jK_{ki} \right)\,dx^{ijk}\,,
\end{eq} 
hence we associate to it the generating function\footnote{With $\circlearrowleft$ we mean the cyclic permutations of $(\nu,\lambda,\mu)$.} $D_\nu K(\lambda,\mu) + \circlearrowleft$. To obtain $\delta dK$, we simply calculate $D_\nu K(\lambda,\mu) + \circlearrowleft$. The same holds for $\delta d V$.
We will need the following identities:
\begin{eqnarray}
&&\Tr Q(\lambda)\delta(D_\nu Q(\mu)) = \Tr\phi(\mu)^{-1} \left( [D_\nu Q(\mu),Q(\lambda)] + D_\nu \phi(\mu) \phi(\mu)^{-1} [Q(\lambda),Q(\mu)]\right) \delta \phi(\mu) \nonumber\\
&&\hspace{3.2cm}- \Tr \phi(\mu)^{-1}[Q(\lambda),Q(\mu)] \delta (D_\nu  \phi(\mu)) \,,\\
&&\Tr D_\nu Q(\lambda) \delta Q(\mu) =\Tr \phi(\mu)^{-1} [Q(\mu),D_\nu Q(\lambda)] \delta \phi(\mu)\,.
\end{eqnarray}
We have that
\begin{eqsplit}
    D_\nu K(\lambda,\mu) =  & \Tr( - \phi(\mu)^{-1} D_\nu \phi(\mu) \phi(\mu)^{-1} D_\lambda \phi(\mu) Q_0 + \phi^{-1}(\mu) D_\nu  D_\lambda \phi(\mu) Q_0\\
    &+ \phi(\lambda)^{-1} D_\nu  \phi(\lambda) \phi(\lambda)^{-1} D_\mu \phi(\lambda) Q_0 - \phi^{-1}(\lambda) D_\nu  D_\mu \phi(\lambda) Q_0)\,.
\end{eqsplit}
We now apply the $\delta$-differential. 
\begin{eqsplit}
    \delta D_\nu  K(\lambda,\mu) 
    =& \Tr(\phi(\mu)^{-1} \big( D_\nu  \phi(\mu) \phi(\mu)^{-1} D_\lambda \phi(\mu) \phi(\mu)^{-1} Q(\mu) \\
    &+ D_\lambda \phi(\mu) \phi(\mu)^{-1} Q(\mu) D_\nu  \phi(\mu) \phi(\mu)^{-1}- D_\nu  D_\lambda \phi(\mu) \phi(\mu)^{-1} Q(\mu)\big) \delta \phi(\mu)\\
    &- \phi(\mu)^{-1} D_\lambda \phi(\mu) \phi(\mu)^{-1} Q(\mu) \delta (D_\nu  \phi(\mu)) \\
    &- \phi(\mu)^{-1} Q(\mu) D_\nu  \phi(\mu) \phi(\mu)^{-1} \delta (D_\lambda \phi(\mu))\\
    & + \phi(\mu)^{-1} Q(\mu) \delta (D_\nu  D_\lambda \phi(\mu)) - (\lambda \leftrightarrow \mu))\,.
\end{eqsplit}
We add the cyclic sum and we select the coefficients of $\delta \phi(\mu)$, $\delta D_\nu  \phi(\mu)$, etc. adding the corresponding terms from $\delta D_\lambda K(\mu,\nu )$.
\begin{eqsplit}
   \delta d K 
    =& \Tr(\phi(\mu)^{-1} \big( D_\mu \phi(\mu) \phi(\mu)^{-1} D_\lambda Q(\mu) - D_\lambda \phi(\mu) \phi(\mu)^{-1} D_\mu Q(\mu)) \delta \phi(\mu)\\
    &- \phi(\mu)^{-1} D_\lambda Q(\mu) \delta (D_\mu \phi(\mu)) + \phi(\mu)^{-1} D_\mu Q(\mu) \delta (D_\lambda\phi(\mu))+ \circlearrowleft)\,.
\end{eqsplit}
We do the same for $V(\lambda,\mu) = \frac{1}{2(\mu-\lambda)} \Tr(Q(\lambda)-Q(\mu))^2$. 
\begin{eqsplit}
    D_\nu V(\lambda,\mu) = & \frac{1}{\mu-\lambda}\Tr ( D_\nu Q(\lambda) - D_\nu Q(\mu))(Q(\lambda) - Q(\mu))\\
    =& \frac{1}{\lambda-\mu}\Tr(D_\nu Q(\lambda) Q(\mu) + D_\nu Q(\mu) Q(\lambda))\,,
\end{eqsplit}
where we used $\Tr Q(\lambda) D_\nu Q(\lambda) =\Tr Q(\mu) D_\nu Q(\mu)=0$. We now apply the $\delta$-differential.
\begin{eq}
    \delta D_\nu V(\lambda,\mu) =  \frac{1}{\lambda-\mu} \Tr \left(D_\nu Q(\lambda) \delta Q(\mu) + Q(\lambda) \delta (D_\nu Q(\mu)) \right) - (\lambda\leftrightarrow \mu)
\end{eq}
and using the identities above we get
\begin{eqsplit}
    \delta D_\nu  V(\lambda,\mu)
 =& \frac{1}{\lambda-\mu} \Tr \Big( - \phi(\mu)^{-1}[Q(\lambda),Q(\mu)] \delta (D_\nu  \phi(\mu)) \\
&+ \phi(\mu)^{-1} \left( [D_\nu Q(\mu),Q(\lambda)] +  [Q(\mu),D_\nu Q(\lambda)] + D_\nu  \phi(\mu) \phi(\mu)^{-1} [Q(\lambda),Q(\mu)] \right) \delta \phi(\mu) \Big)\\
& - (\lambda \leftrightarrow \mu)\,.
\end{eqsplit}
We add the cyclic sum and we select the coefficients of $\delta \phi(\mu)$, $\delta D_\nu  \phi(\mu)$, etc. adding the corresponding terms from $\delta D_\lambda V(\mu,\nu )$.
\begin{eqsplit}
    \delta d V = & \frac{1}{\lambda-\mu} \Tr \Big( - \phi(\mu)^{-1}[Q(\lambda),Q(\mu)] \delta (D_\nu  \phi(\mu)) \\
&+ \phi(\mu)^{-1} \left( [D_\nu Q(\mu),Q(\lambda)] +  [Q(\mu),D_\nu Q(\lambda)] + D_\nu  \phi(\mu) \phi(\mu)^{-1} [Q(\lambda),Q(\mu)] \right) \delta \phi(\mu) \Big)\\
&+ \frac{1}{\nu -\mu} \Tr \Big(  \phi(\mu)^{-1}[Q(\nu ),Q(\mu)] \delta (D_\lambda \phi(\mu)) \\
&- \phi(\mu)^{-1} \left( [D_\lambda Q(\mu),Q(\nu )] +  [Q(\mu),D_\lambda Q(\nu )] + D_\lambda  \phi(\mu) \phi(\mu)^{-1} [Q(\nu ),Q(\mu)] \right) \delta \phi(\mu) \Big) + \circlearrowleft\,.
\end{eqsplit}
By comparing the coefficients of $\delta D_\nu  \phi(\mu)$ and of $\delta D_\lambda  \phi(\mu)$ we get the desired equations \eqref{eq_Q}. The equations coming from the coefficients of $\delta \phi(\mu)$ are differential consequences of them.

We turn to the closure relation. We are going to use the following identities:
\begin{gather}
    \frac{1}{\mu - \nu } - \frac{1}{\lambda  - \nu } = \frac{\lambda  - \mu}{(\mu - \nu )(\lambda  - \nu )}\,,\\
    \frac{1}{(\mu - \nu )(\lambda -\nu )} + \frac{1}{(\nu  - \lambda )(\mu - \lambda )} + \frac{1}{(\nu -\mu)(\lambda -\mu)}=0\,,\\
    \Tr [Q(\lambda ),Q(\mu)]Q(\lambda ) =0\,,\\
    \Tr[Q(\lambda ),Q(\nu )]Q(\mu) =  \Tr[Q(\mu),Q(\lambda )]Q(\nu )\,.
\end{gather}
A direct computation shows that the kinetic term brings
\begin{eqsplit}
   & D_\nu  K(\lambda ,\mu) + D_\lambda  K(\mu,\nu ) + D_\mu K (\nu ,\lambda ) = \Tr(D_\nu  \phi(\lambda ) \phi(\lambda )^{-1} D_\mu Q(\lambda ) +\circlearrowleft)\\
   &=\Tr(\frac{1}{\lambda -\mu} D_\nu Q(\lambda )Q(\mu) + \circlearrowleft)\\
   &=\Tr(\frac{1}{(\lambda -\mu)(\lambda -\mu)}[Q(\lambda ),Q(\nu )]Q(\mu) +\circlearrowleft)\\
   & =\Tr\left( \left(\frac{1}{(\lambda -\mu)(\lambda -\mu)}  +\circlearrowleft \right)[Q(\lambda ),Q(\nu )]Q(\mu)\right) = 0
\end{eqsplit}
We have that the potential term brings:
\begin{eqsplit}
 D_\nu  V(\lambda ,\mu) &= - \frac{1}{2(\lambda -\mu)} 2 \Tr (D_\nu  Q(\lambda ) - D_\nu Q(\mu))(Q(\lambda )-Q(\mu))\\
 &= \frac{1}{\mu-\lambda } \Tr \frac{[Q(\lambda ),Q(\nu )]}{\lambda -\nu } (Q(\lambda ) - Q(\mu)) - \frac{1}{\mu-\lambda } \Tr \frac{[Q(\mu),Q(\nu )]}{\mu-\nu }(Q(\lambda )-Q(\mu))\\
 &=\frac{1}{(\lambda -\mu)(\lambda -\nu )} \Tr[Q(\lambda ),Q(\nu )]Q(\mu)  + \frac{1}{(\lambda -\mu)(\mu-\nu )} \Tr [Q(\mu),Q(\nu )]Q(\lambda )\\
 &=\frac{1}{\lambda -\mu} \left( \frac{1}{\lambda -\nu } - \frac{1}{\mu-\nu }\right) \Tr[Q(\lambda ),Q(\nu )]Q(\mu)\\
 & = \frac{1}{(\lambda -\nu )(\nu -\mu)} \Tr[Q(\lambda ),Q(\nu )]Q(\mu)
\end{eqsplit}
The cyclic sum then is
\begin{eq}
\left(\frac{1}{(\lambda -\nu )(\nu -\mu)} + \frac{1}{(\nu -\mu)(\mu-\lambda )} + \frac{1}{(\mu-\lambda )(\lambda -\nu )} \right)\Tr [Q(\mu),Q(\lambda )]Q(\nu ) = 0
\end{eq}
so that we have that the Lagrangian multiform satisfies the closure relation $d \Lag =0$.
\end{proof}
\subsection{Proof of Proposition \ref{prop_symplecticmultiform}}\label{proof_symplecticmultiform}
\begin{proof}
First, we claim that $\Omegaone$ is given by the generating function 
\begin{eq}
\Omegaone(\lambda) = \Tr\left( Q_0 \phi(\lambda)^{-1} \delta \phi(\lambda)\right)\,.
\end{eq}
We need to show that $\delta \Lag + d \Omegaone = 0$ on the multiform Euler-Lagrange equations $D_\mu Q(\lambda) =\frac{[Q(\mu),Q(\lambda)]}{\mu-\lambda}$.
For convenience, let us denote $\psi(\lambda) \coloneqq \phi^{-1}(\lambda)$. A direct computation shows that
\begin{eqsplit}
    \delta K(\lambda,\mu) =& \Tr\bigg( D_\lambda \phi(\mu) Q_0 \delta \psi(\mu) + Q_0 \psi(\mu) \delta (D_\lambda \phi(\mu))\\
    &- D_\mu \phi(\lambda) Q_0 \delta \psi(\lambda) - Q_0 \psi(\lambda) \delta (D_\mu \phi(\lambda))\bigg)\,,
\end{eqsplit}
and
\begin{eq}
    \delta V(\lambda,\mu) =  \Tr \bigg( \frac{1}{\lambda-\mu}\psi(\lambda) [Q(\lambda),Q(\mu)] \delta\phi(\lambda)  - \frac{1}{\lambda-\mu}\psi(\mu)[Q(\lambda),Q(\mu)]\delta \phi(\mu)\bigg)\,.
\end{eq}
The coefficient of the generating function $\displaystyle \Omegaone(\lambda) = \sum_{k=0}^\infty \omegaone_k/\lambda^{k+1}$ are obtained as (note that $\omegaone_0=0$) 
\begin{eq}
\omegaone_k = \Tr \sum_{i=1}^k Q_0 \psi_i \delta \phi_{k+1-i}\,.
\end{eq}
Hence, for the corresponding form, we have using the variational bicomplex calculus,
\begin{eqsplit}
    d\Omegaone =&d\left(\sum_{k=1}^\infty \omegaone_k \wedge dx^k  \right)\\
    =& \Tr\sum_{j<k=1}^\infty \sum_{i=1}^k\left( \partial_k \psi_i \delta \phi_{j+1-i} + \psi_i\delta (\partial_k \phi_{j+1-i}) - \partial_j \psi_i \delta \phi_{k+1-i} - \psi_i \delta (\partial_j \phi_{k+1-i}) \right) \wedge dx^{jk}\,.
\end{eqsplit}
The associated generating function is given by
\begin{eqsplit}
d\Omegaone(\lambda,\mu) = \Tr \bigg(& Q_0 D_\mu \psi(\lambda) \delta \phi(\lambda) + Q_0 \psi(\lambda) \delta (D_\mu \phi(\lambda))\\
&- Q_0 D_\lambda \psi(\mu) \delta \phi(\mu) -Q_0 \psi(\mu) \delta(D_\lambda \phi(\mu)) \bigg)\,.
\end{eqsplit}
So the sum $\delta K(\lambda,\mu) - \delta V(\lambda,\mu) + d\Omegaone(\lambda,\mu)$ reads 
\begin{eqsplit}
    \Tr \bigg(&\psi(\lambda) D_\mu Q(\lambda) \delta \phi(\lambda)- \frac{1}{\lambda-\mu} \psi(\lambda)[Q(\lambda),Q(\mu)] \delta \phi(\lambda) \\
   &-\psi(\mu) D_\lambda Q(\mu) + \frac{1}{\lambda-\mu}\psi(\lambda) [Q(\lambda),Q(\mu)] \delta \phi(\mu)\bigg)\,.
\end{eqsplit}
This vanishes on the multiform Euler-Lagrange equations 
\begin{eq}
D_\mu Q(\lambda) =\frac{[Q(\mu),Q(\lambda)]}{\mu-\lambda}\,,\qquad D_\lambda Q(\mu) =\frac{[Q(\lambda),Q(\mu)]}{\lambda-\mu}\,,
\end{eq}
thus completing the argument. As a consequence,
\begin{eq}
\Omega(\lambda)=\delta \Omegaone(\lambda)=- \Tr \left(Q_0 \phi(\lambda)^{-1} \delta \phi(\lambda) \wedge \phi(\lambda)^{-1} \delta \phi(\lambda)\right)\,,
\end{eq}
as required.
\end{proof}
\subsection{Proof of Proposition \ref{Hamiltonian1forms}}\label{proof_Hamiltonian1forms}
\begin{proof}
We start with the general expression of the vertical vector field
\begin{equation}
    \xi_F = \sum_{j=1}^\infty \left( A_j \partial_{f_j}+B_j \partial_{e_j}   \right)\,,
\end{equation}
and determine $A_j$, $B_j$ such that $\ip{\xi_F}{\Omega} = \delta F$ holds, or equivalently,
\begin{equation}
    \ip{\xi_F}{\omega_k} = \delta F_k\,, \qquad \forall k\ge 0 \,.
\end{equation}
Since $\omega_0=0$ we instantly get that $F_0$ has to be constant. The left-hand side reads
\begin{eqsplit}
    \ip{\xi_F}{\omega_k} &=  \sum_{i=1}^k \sum_{j=1}^\infty (A_j \delta_{i,j} \delta e_{k-i+1} - B_j \delta_{j,k-i+1} \delta f_i)\\
    &= \sum_{i=1}^k(A_{k-i+1} \delta e_i - B_{k-i +1} \delta f_i)\,,
\end{eqsplit}
whilst the right hand-side is
\begin{equation}
    \sum_{i=1}^\infty \left( \parder{F_k}{e_i} \delta e_i + \parder{F_k}{f_i} \delta f_i \right)\,.
\end{equation}
Comparing the two we get
\begin{subgather}
    \parder{F_k}{e_i} = \parder{F_k}{f_i}=0\,, \qquad \forall i>k\,,\\
    \parder{F_k}{e_i} = A_{k-i+1}\,,\quad \parder{F_k}{f_i} = -B_{k-i+1}\,,\qquad \forall i\le k\,.
\end{subgather}
The latter brings that
\begin{equation}
    \parder{F_k}{e_i} = A_{k-i+1} = A_{(k+1) -(i+1) +1} = \parder{F_{k+1}}{e_{i+1}}\,,
\end{equation}
and similarly for $f_i$. These conditions are necessary and sufficient. 
\end{proof}
\subsection{Proof of Proposition \ref{Hamiltonian0forms}}\label{proof_Hamiltonian0forms}
\begin{proof}
We need to show that 
\begin{equation}
    \ip{\xi_H}{\Omega} = \delta H \quad\text{where}\quad \Omega = \sum_{k=1}^\infty \sum_{m=1}^k \delta f_m \wedge \delta e_{k+1-m} \wedge d x^k\,.
\end{equation}
We start with the left hand-side
\begin{eqsplit}
    \ip{\xi_H}{\Omega} =& \sum_{i=1}^\infty  \sum_{k=1}^\infty \sum_{m=1}^k \ip{\left( - \parder{H}{f_i} \partial_{e_1} \wedge \partial_i +  \parder{H}{e_i} \partial_{f_1} \wedge \partial_i \right)}{\left( \delta f_m \wedge \delta e_{k+1-m} \wedge \delta x^k\right)}\\
   = &  \sum_{i=1}^\infty  \sum_{k=1}^\infty \sum_{m=1}^k  \left( \parder{H}{f_i} \delta_{ik} \delta_{k+1-m,1} \delta f_m \right) + \sum_{i=1}^\infty  \sum_{k=1}^\infty \sum_{m=1}^k  \left( \parder{H}{e_i} \delta_{ik} \delta_{m,1} \delta e_{k+1-m} \right)\\
   =& \sum_{i=1}^\infty \parder{H}{f_i} \delta f_i + \sum_{i=1}^\infty \parder{H}{e_i} \delta e_i = \delta H \,.\qquad
\end{eqsplit}
\end{proof}

\subsection{Proof of Theorem \ref{rmatrix}}\label{proof_rmatrix}
\begin{lmm}
\label{lemmaPB}
For each $k\ge 0$, the only non-zero single-time Poisson of $a_i$, $b_i$ and $c_i$, $0\le i \le k$, are given by
\begin{gather}
    \pb{a_i}{b_j}_k = b_{i+j-k-1}\,,\\
    \pb{a_i}{c_j}_k = - c_{i+j-k-1}\,,\\
    \pb{b_i}{c_j}_k = 2a_{i+j-k-1} \,.
\end{gather}
For convenience, we use the convention that a coefficient in a series vanishes when its index is negative. Hence, it is understood that $\pb{a_i}{b_j}_k = \pb{a_i}{c_j}_k = \pb{b_i}{c_j}_k=0$ whenever $i+j <k+1$.  
\end{lmm}
\begin{proof}
 We start with the fact that for any power series $\alpha$ and $\beta$ we have
 \begin{eq}
 \label{formula_sum}
     \sum_{\ell=1}^k \alpha_{i-\ell} \beta_{j+\ell-k-1} = (\alpha\beta)_{i+j-k-1}\,.
 \end{eq}
 In fact, by limiting the sum only to the non-zero terms:
 \begin{eq}
     \sum_{\ell=1}^k \alpha_{i-l} \,\beta_{j+\ell-k-1} = \sum_{\ell=k+1-j}^i \alpha_{i-\ell}\,\beta_{j+\ell-k-1} = \sum_{m=0}^{i+j-k-1} \alpha_{i+j-k-1-m}\,\beta_m = (\alpha\beta)_{i+j-k-1}\,. 
 \end{eq}
 We study the case where $k+1-j\le j$, namely $i+j \ge k+1$. If $i+j <k+1$ then the sum is empty, and the result is zero. We are now ready to compute the following Poisson brackets using the formulas in Appendix \ref{efcoordinates_section}.
 \begin{eqsplit}
     \pb{a_i}{b_j}_k =& \sum_{\ell=1}^k \left(  \parder{a_i}{f_\ell} \parder{b_j}{e_{k+1-\ell}} -\parder{b_j}{f_\ell} \parder{a_i}{e_{k+1-\ell}} \right)\\
    =&  \sum_{\ell=1}^k \left(  e_{i-\ell} \left(\frac{i-3a}{2\sqrt{i-a}}\right)_{j+\ell-k-1} -\left(\frac{-e^2}{2\sqrt{i-a}}\right)_{j-\ell}f_{i+\ell-k-1} \right)\\
    =&  \left( \frac{ie-3ae +e^2f }{2\sqrt{i-a}} \right)_{i+j-k-1}  = \left( \frac{ie -3ae + (i+a)e}{2\sqrt{i-a}}\right)_{i+j-k-1}\\
    =& (e\sqrt{i-a})_{i+j-k-1} = b_{i+j-k-1}\,.
 \end{eqsplit}
 \begin{eqsplit}
     \pb{a_i}{c_j}_k =& \sum_{\ell=1}^k \left(  \parder{a_i}{f_\ell} \parder{c_j}{e_{k+1-\ell}} -\parder{c_j}{f_\ell} \parder{a_i}{e_{k+1-\ell}} \right)\\
    =&  \sum_{\ell=1}^k \left(  e_{i-\ell} \left(\frac{-f^2}{2\sqrt{i-a}}\right)_{j+\ell-k-1} -\left(\frac{i-3a}{2\sqrt{i-a}}\right)_{j-\ell}f_{i+\ell-k-1} \right)\\
    =&  \left( \frac{-ef^2 -f(i-3a)}{2\sqrt{i-a}} \right)_{i+j-k-1}  = \left( \frac{-f(i+a) - if+3af}{2\sqrt{i-a}}\right)_{i+j-k-1}\\
    =& -(f\sqrt{i-a})_{i+j-k-1} =- c_{i+j-k-1}\,.
 \end{eqsplit}
\begin{eqsplit}
     \pb{a_i}{a_j}_k =& \sum_{\ell=1}^k \left(  \parder{a_i}{f_\ell} \parder{a_j}{e_{k+1-\ell}} -\parder{a_j}{f_\ell} \parder{a_i}{e_{k+1-\ell}} \right)\\
    =&  \sum_{\ell=1}^k \left(  \left( \frac{-e^2}{2\sqrt{i-a}} \right)_{i-\ell} \left(\frac{-f^2}{2\sqrt{i-a}}\right)_{j+\ell-k-1} -\left(\frac{i-3a}{2\sqrt{i-a}}\right)_{j-\ell}\left(\frac{i-3a}{2\sqrt{i-a}}\right)_{i+\ell-k-1} \right)\\
    =&  \left( \frac{e^2f^2 - (i-3a)^2 }{4(i-a)} \right)_{i+j-k-1}  = \left( \frac{(i+a)^2 - (i-3a)^2 }{4(i-a)}\right)_{i+j-k-1}\\
    =&2a _{i+j-k-1}\,.
 \end{eqsplit}
\end{proof}
\begin{remark}
These Poisson bracket coincides with the $\pb{~}{~}_{-k}$ in \cite{Avan_Caudrelier_2017}. In this instance we don't take the Poisson brackets of $a_i,b_i,c_i$ for $i>k$ because they do not belong to the $k$-th single-time phase space.
\end{remark}

\begin{proof}[Proof of Theorem \ref{rmatrix}]
We start by proving that
\begin{equation}
    \parder{Q}{e_k}(\lambda) = \lambda \parder{Q}{e_{k+1}}(\lambda)\,, \qquad \parder{Q}{f_k}(\lambda) = \lambda \parder{Q}{f_{k+1}}(\lambda)\,.
\end{equation}
This is done for each matrix element. In fact
\begin{eq}
    \parder{Q}{e_k} = \parder{Q}{e} \parder{e}{e_k} = \lambda^{-k} \parder{Q}{e} = \lambda \lambda ^{-k-1}\parder{Q}{e} = \lambda \parder{Q}{e}\parder{e}{e_{k+1}} = \lambda \parder{Q}{e_{k+1}}\,.
\end{eq}
Similarly $\parder{Q}{f_k} = \lambda \parder{Q}{f_{k+1}}$. By virtue of this result, and since $Q_0$ is constant, we have the following:
\begin{subgather}
\label{ham_prop_Q}
     \parder{Q}{e_k} = \sum_{j=0}^\infty \parder{Q_j}{e_k} \lambda^{-j} \\
     \lambda \parder{Q}{e_{k+1}} = \lambda \sum_{i=0}^\infty \parder{Q_i}{e_{k+1}} \lambda^{-i} = \sum_{i=0}^\infty \parder{Q_i}{e_{k+1}} \lambda^{-i+1} = \sum_{j=0}^\infty \parder{Q_{j+1}}{e_{k+1}} \lambda^{-j}\,.
\end{subgather}
If we look at the coefficients in $\lambda$ we see that, for all $j$ and $k$, $\parder{Q_j}{e_k} = \parder{Q_{j+1}}{e_{k+1}}$. Similarly one can obtain that $\parder{Q_j}{f_k} = \parder{Q_{j+1}}{f_{k+1}}$.\\
Finally, we check that the Lax form is Hamiltonian, using Proposition \ref{Hamiltonian1forms}, i.e. that 
\begin{eqsplit}
    \parder{Q^{(i)}}{e_j} &= \sum_{k=0}^i \lambda^{i-k} \parder{Q_k}{e_j}= \sum_{k=0}^i \lambda^{i-k} \parder{Q_{k+1}}{e_{j+1}}=  \sum_{k=0}^i \lambda^{(i+1)-(k+1)} \parder{Q_{k+1}}{e_{j+1}}\\
    &= \sum_{k=1}^{i+1} \lambda^{i+1-k} \parder{Q_k}{e_{j+1}} =\parder{Q^{(i+1)}}{e_{j+1}} - \lambda^{i+1} \parder{Q_0}{e_{j+1}} =  \parder{Q^{(i+1)}}{e_{j+1}}\,,
\end{eqsplit}
where we used that $Q_0$ is constant. Similarly $\parder{Q^{(i)}}{f_j}= \parder{Q^{(i+1)}}{f_{j+1}}$. \\

We now turn to the proof of \eqref{PBLax}. Thanks to the decomposition of the multi-time Poisson bracket into single-time Poisson brackets, we have that  $\cpb{W_1(\lambda)}{W_2(\mu)} = [r_{12}(\lambda,\mu),W_1(\lambda) + W_2(\mu)]$ if and only for all $k\ge 0$,
\begin{eq}\label{rmatrix_singletime}
    \pb{Q^{(k)}_1(\lambda)}{Q^{(k)}_2(\mu)}_k = [r_{12}(\lambda,\mu),Q^{(k)}_1(\lambda)+Q^{(k)}_2(\mu)]\,.
\end{eq}
Writing $Q^{(k)}(\lambda)=Q^{(k)}_+(\lambda)\sigma_++Q^{(k)}_-(\lambda)\sigma_-+Q^{(k)}_3(\lambda)\sigma_3$, the right hand-side of \eqref{rmatrix_singletime} reads
\begin{eqsplit}
\label{RHS}
    &[r_{12}(\lambda,\mu),Q^{(k)}_1(\lambda)+Q^{(k)}_2(\mu) ]=\frac{2}{\mu-\lambda} (Q^{(k)}_3(\mu) -Q^{(k)}_3(\lambda) ) (\sigma_+ \otimes \sigma_- - \sigma_- \otimes \sigma_+) \\
    &+ \frac{Q^{(k)}_+(\mu) -Q^{(k)}_+(\lambda)}{\mu-\lambda}(\sigma_3 \otimes \sigma_+ - \sigma_+ \otimes \sigma_3) + \frac{Q^{(k)}_-(\lambda) -Q^{(k)}_-(\mu)}{\mu-\lambda}(\sigma_3 \otimes \sigma_- - \sigma_- \otimes \sigma_3) \,,
\end{eqsplit}
while the left hand-side is given by
\begin{eqsplit}
\label{LHS}
    \pb{Q^{(k)}_1(\lambda)}{Q^{(k)}_2(\mu)}_k = &\sum_{i,j=0}^k \frac{(\lambda \mu)^k}{\lambda^i \mu^j} \big(\pb{a_i}{a_j}_k \sigma_3 \otimes \sigma_3 +\pb{b_i}{b_j}_k \sigma_+ \otimes \sigma_+ + \pb{c_i}{c_j}_k \sigma_- \otimes \sigma_-\\
    &+\pb{b_i}{c_j}_k \sigma_+ \otimes \sigma_- + \pb{c_i}{b_j}_k \sigma_- \otimes \sigma_+ + \pb{a_i}{b_j}_k \sigma_3 \otimes \sigma_+\\
    &+ \pb{b_i}{a_j}_k \sigma_+ \otimes \sigma_3 + \pb{a_i}{c_j}_k \sigma_3 \otimes \sigma_- + \pb{c_i}{a_j}_k \sigma_- \otimes \sigma_3 \big)\,.
\end{eqsplit}
We now invoke Lemma \ref{lemmaPB} which gives the necessary single-time Poisson brackets and allows us to check directly that \eqref{RHS} is equal to \eqref{LHS}. We show it for 
the $\sigma_+ \otimes \sigma_-$ component, as the others are obtained similarly. In the left hand-side we have
\begin{equation}
    \frac{2}{\mu - \lambda} ( Q^{(k)}_3(\mu) - Q^{(k)}_3(\lambda) ) = \frac{2}{\mu - \lambda} \sum_{j=0}^k (\mu^{k-j} - \lambda^{k-j}) a_j = 2 \sum_{j=0}^k \sum_{i=1}^{k-j-1} \lambda^i \mu^{k-i-j-1} a_j\,,
\end{equation}
while right hand-side is equal to
\begin{eq}
 2   \sum_{i,j=0}^k \frac{(\lambda \mu)^k}{\lambda^i \mu^j} a_{i+j-k-1} = 2 \sum_{i=0}^k \sum_{m=0}^{i-1} \frac{a_m}{\lambda^{i-k} \mu^{m+1-i}} = 2 \sum_{n=0}^k \sum_{m=0}^{k-n-1} \lambda^n \mu^{k-n-m-1}a_m\,.
\end{eq}
This concludes the proof.
\end{proof}

\subsection{Proof of Theorem \ref{HamiltonianZCE}}\label{proof_HamiltonianZCE}
\begin{proof}
Note the set of zero curvature equations can be written as
\begin{equation}
    d W(\lambda) = W(\lambda) \wedge W(\lambda)\,,
\end{equation}
where the right-hand side is understood as 
\begin{equation}
    W(\lambda) \wedge W(\lambda) = \left( \sum_{i=0}^\infty Q^{(i)}(\lambda) dx^i\right) \wedge \left( \sum_{j=0}^\infty Q^{(j)}(\lambda) dx^j\right) = \sum_{i<j}[Q^{(i)}(\lambda),Q^{(j)}(\lambda)]\, dx^{ij} \,,
\end{equation}
and the left-hand side is of course $dW(\lambda) = \sum_{i<j} (\partial_i Q^{(j)}(\lambda) - \partial_j Q^{(i)}(\lambda)) \,dx^{ij}$. Thus, we have to prove that 
\begin{equation}
\label{goal1}
 W(\lambda) \wedge W(\lambda)=   \sum_{i<j} \cpb{H_{ij}}{W(\lambda)}\,dx^{ij} .
\end{equation}
By definition $\sum_{i<j} \cpb{H_{ij}}{W(\lambda)}= \ip{\xi_{W}(\lambda)}{\delta \Ham} = \sum_{i<j} ( \ip{\xi_{W}(\lambda)}{\delta H_{ij}} )\,dx^{ij}$, where, using the expression \eqref{Lax_vector_field} for $\xi_W(\lambda)$, we find
\begin{equation}
 \ip{\xi_{W}(\lambda)}{\delta H_{ij}}=\sum_{k=1}^j \left( \parder{Q^{(k)}(\lambda)}{e_1} \parder{H_{ij}}{f_k} - \parder{Q^{(k)}(\lambda)}{f_1} \parder{H_{ij}}{e_k}\right)\,.
\end{equation}
Hence \eqref{goal1} is equivalent to, for $i<j$,
\begin{equation}\label{constructionHij}
 [Q^{(i)}(\lambda),Q^{(j)}(\lambda)] = \sum_{k=1}^j \left( \parder{Q^{(k)}(\lambda)}{e_1} \parder{H_{ij}}{f_k} - \parder{Q^{(k)}(\lambda)}{f_1} \parder{H_{ij}}{e_k}\right)  \,.
\end{equation}
We prove the latter in generating form as follows. 
We multiply both sides by $\mu ^{-i-1}\nu ^{-j-1}$ and form the following sums over $i$ and $j$
\begin{equation}
  \sum_{j=0}^\infty \sum_{i=0}^j \frac{1}{\mu ^{i+1}\nu ^{j+1}}  [Q^{(i)}(\lambda),Q^{(j)}(\lambda)] =   \sum_{j=0}^\infty \sum_{i=0}^j  \frac{1}{\mu ^{i+1}\nu ^{j+1}}\sum_{k=1}^j \left( \parder{Q^{(k)}(\lambda)}{e_1} \parder{H_{ij}}{f_k} - \parder{Q^{(k)}(\lambda)}{f_1} \parder{H_{ij}}{e_k}\right)  \,.
\end{equation}
We can rearrange the sums in the right-hand side to get
\begin{equation}
     \sum_{k=1}^\infty \sum_{j=k}^\infty \sum_{i=0}^j  \frac{1}{\mu ^{i+1}\nu ^{j+1}} \left( \cdots\right)
     =
     \sum_{k=1}^\infty \sum_{j=0}^\infty \sum_{i=0}^j  \frac{1}{\mu ^{i+1}\nu ^{j+1}} \left( \cdots\right)
     \,,
\end{equation}
where we have used the fact that $H_{ij}$ depends only on $e_1,\dots,e_j$ and $f_1,\dots,f_j$ in the second step to extend the sum over $j$ from $0$ instead of $k$. We can similarly form the sums with $\mu \leftrightarrow \nu $ and use the same trick to rearrange the sums in the right-hand side. Using the anti-symmetry of both left and right-hand side of \eqref{constructionHij}, we come to the following generating form of \eqref{constructionHij}
\begin{eqsplit}
        &  \sum_{i,j=0}^\infty \frac{1}{\mu ^{i+1}\nu ^{j+1}}  [Q^{(i)}(\lambda),Q^{(j)}(\lambda)] \\
    =& \sum_{k=1}^\infty \left( \parder{Q^{(k)}(\lambda)}{e_1} \parder{}{f_k}\sum_{i,j=0}^\infty\frac{H_{ij}}{\mu ^{i+1}\nu ^{j+1}} - \parder{Q^{(k)}(\lambda)}{f_1} \parder{}{e_k}\sum_{i,j=0}^\infty\frac{H_{ij}}{\mu ^{i+1}\nu ^{j+1}}\right)\,,
\end{eqsplit}
\ie,
\begin{equation}\label{constructionH:generating}
    \frac{[Q(\mu ),Q(\nu )]}{(\mu -\lambda)(\nu -\lambda)}= \sum_{k=1}^\infty \parder{Q^{(k)}(\lambda)}{e_1} \parder{\Ham(\mu ,\nu )}{f_k} - \sum_{k=1}^\infty\parder{Q^{(k)}(\lambda)}{f_1} \parder{\Ham(\mu ,\nu )}{e_k} \,,
\end{equation}
where we have used 
\begin{equation}
    \sum_{i=0}^\infty \frac{Q^{(i)}(\lambda)}{\mu ^{i+1}} = \frac{Q(\mu )}{\mu -\lambda}\,.
\end{equation}
We now show that \eqref{constructionH:generating} holds by computing its right-hand side recalling that 
\begin{equation}
    \Ham(\mu ,\nu )= \frac{1}{2(\nu -\mu )} \Tr (Q(\mu )-Q(\nu ))^2\,.
\end{equation}
For convenience, denote $a(\mu ), a(\nu ), a(\lambda)$ by $ a,  a', a''$ respectively and similarly for $b$ and $c$. We have
\begin{equation}
    \parder{\Ham(\mu ,\nu )}{f_k} = \frac{1}{(\nu -\mu )}\Tr{\left( \parder{Q(\mu )}{f_k} - \parder{Q(\nu )}{f_k}\right)(Q(\mu )-Q(\nu ))}
\end{equation}
and 
\begin{eqsplit}
&\frac{1}{\nu -\mu } \sum_{k=1}^\infty \parder{Q^{(k)}(\lambda)}{e_1} \Tr \parder{Q(\mu )}{f_k} (Q(\mu )- Q(\nu ))\\
&=  \frac{1}{\nu -\mu } \sum_{k=1}^\infty \parder{Q^{(k)}(\lambda)}{e_1} \frac{1}{\mu ^k \sqrt{i-a}} \Tr \begin{pmatrix}
    b & - \frac{b^2}{2(i-a)}\\
    \frac{i-3a}{2} &-b
    \end{pmatrix}
    \begin{pmatrix}
    a-a'&b-b'\\
    c-c'& a'-a
    \end{pmatrix}\\
    &= \frac{1}{\nu -\mu } \sum_{k=1}^\infty \parder{Q^{(k)}(\lambda)}{e_1} \frac{1}{\mu ^k \sqrt{i-a}} \left(2b (a-a') - \frac{b^2(c-c')}{2(i-a)} + \frac{(i-3a)(b-b')}{2}\right)\\
    & = \frac{\mu }{(\nu -\mu )}  \parder{}{e_1} \frac{Q(\mu )}{\mu -\lambda}\frac{1}{ \sqrt{i-a}} \left(2b (a-a') - \frac{b^2(c-c')}{2(i-a)} + \frac{(i-3a)(b-b')}{2}\right)\\
    & = \frac{1}{(\nu -\mu )(\mu -\lambda)} \frac{1}{ i-a}\left(2b (a-a') - \frac{b^2(c-c')}{2(i-a)} + \frac{(i-3a)(b-b')}{2}\right) \begin{pmatrix}
    c & \frac{i-3a}{2}\\
    - \frac{c^2}{2(i-a)} & -c    
    \end{pmatrix}\,,
\end{eqsplit}
where we have used that $\displaystyle\sum_{k=1}^\infty \frac{Q^{(k)}}{\mu ^k} =\mu ( \frac{Q(\mu )}{\mu -\lambda} - \frac{Q_0}{\mu })$ and that $Q_0$ is constant.
 Similarly, we have 
\begin{eqsplit}
     &\frac{1}{\nu -\mu } \sum_{k=1}^\infty \parder{Q^{(k)}(\lambda)}{e_1} \Tr \parder{Q(\nu )}{f_k} (Q(\mu )- Q(\nu )) =\\
     & \frac{1}{(\nu -\mu )(\nu -\lambda)} \frac{1}{ i-a'}(2b' (a-a') - \frac{{b'}^2(c-c')}{2(i-a')} + \frac{(i-3a')(b-b')}{2}) \begin{pmatrix}
    c' & \frac{i-3a'}{2}\\
    - \frac{{c'}^2}{2(i-a')} & -c'    
    \end{pmatrix}\,,
\end{eqsplit}
\begin{eqsplit}
    &\frac{1}{\nu -\mu } \sum_{k=1}^\infty \parder{Q^{(k)}(\lambda)}{f_1} \Tr \parder{Q(\mu )}{e_k} (Q(\mu )- Q(\nu )) =\\
     & \frac{1}{(\nu -\mu )(\mu -\lambda)} \frac{1}{ i-a}(2c (a-a') - \frac{{c}^2(b-b')}{2(i-a)} + \frac{(i-3a)(c-c')}{2}) \begin{pmatrix}
    b &- \frac{{b}^2}{2(i-a)}\\
      \frac{i-3a}{2}& -b    
    \end{pmatrix}\,,
\end{eqsplit}
and
\begin{eqsplit}
    & \frac{1}{\nu -\mu } \sum_{k=1}^\infty \parder{Q^{(k)}(\lambda)}{f_1} \Tr \parder{Q(\nu )}{e_k} (Q(\mu )- Q(\nu ))=\\
     & \frac{1}{(\nu -\mu )(\nu -\lambda)} \frac{1}{ i-a'}(2c' (a-a') - \frac{{c'}^2(b-b')}{2(i-a')} + \frac{(i-3a')(c-c')}{2}) \begin{pmatrix}
    b' &- \frac{{b'}^2}{2(i-a')}\\
      \frac{i-3a'}{2}& -b'    
    \end{pmatrix}\,.
\end{eqsplit}
We collect all the contributions on the $\sigma_3$ component for instance (the other two are obtained similarly).  The numerator of $\frac{N_1}{(\nu -\mu )(\mu -\lambda)(i-a)}$ is
\begin{eqsplit}
    &N_1=2bc(a-a') - \frac{cb^2(c-c')}{2(i-a)} + \frac{(i-3a)(cb-cb')}{2} - 2cb(a-a') + \frac{bc^2(b-b')}{2(i-a)} - \frac{(i-3a)(bc-bc')}{2} \\
    =&\frac{1}{2(i-a)} (-c^2 b^2 + cc'b^2 + b^2c^2 - bb'c^2) + \frac{i-3a}{2}(cb-cb'-bc+bc')\\
    =&\frac{bc}{2(i-a)}(bc'-b'c) + \frac{i-3a}{2} (bc'-cb')\\
    =&(i-a)(bc'-b'c)
\end{eqsplit}
where in the last equality, we have used that $bc = -1-a^2=(i-a)(i+a)$. Similarly, the numerator of $\frac{N_2}{(\nu -\mu )(\nu -\lambda)(i-a')}$ is $-(i-a')(bc'-b'c)$, by simply swapping $\mu $ and $\nu $. So, in total the $\sigma_3$ component of the right-hand side of \eqref{constructionH:generating} is given by
\begin{equation}
    \frac{bc'-b'c}{\nu -\mu } \left( \frac{1}{\mu -\lambda} - \frac{1}{\nu -\lambda} \right) = \frac{bc'-b'c}{(\mu -\lambda) (\nu -\lambda)}\,.
\end{equation}
This is exactly the coefficient of the $\sigma_3$ component of $\frac{[Q(\mu ),Q(\nu )]}{(\mu -\lambda)(\nu -\lambda)}$ as is readily seen. The other components are dealt with in the same way, and are omitted for brevity.
\end{proof}

\addcontentsline{toc}{chapter}{References}
\renewcommand\bibname{References}
\printbibliography

\end{document}